\documentclass[12pt]{article}
\usepackage{latexsym,graphicx,amssymb,color,amsmath}
\usepackage{relsize,ccaption,url,latexsym,epsfig,a4wide}
\usepackage{lscape}
\usepackage{pdfsync}
\usepackage[numbers,sort&compress]{natbib}
\usepackage{hyperref}
\usepackage{cancel}
\usepackage{epstopdf}
\usepackage{color}
\definecolor{cherry}{rgb}{0.9,.1,.2}

\newcommand{\br}[1]{\boldsymbol{\textcolor{cherry}{#1}}}

\newcommand{\bl}[1]{\boldsymbol{\textcolor{black}{#1}}}
\newcommand\nop[1]{\mathop{:\!#1\!\!:}}
\DeclareMathOperator{\ch}{ch}
\DeclareGraphicsExtensions{.eps}

\numberwithin{equation}{section}
\newcommand{\be}{\begin{equation}}
\newcommand{\ee}{\end{equation}}
\newcommand{\ba}{\begin{eqnarray}}
\newcommand{\ea}{\end{eqnarray}}
\newcommand{\hf}{\frac{1}{2}}
 
\newcommand\req[1]{(\ref{#1})}

\newcommand\F{\mathbb{F}}
\newcommand\C{\mathbb{C}}

\newcommand\N{\mathbb{N}}
\renewcommand\P{\mathbb{P}}

\newcommand\R{\mathbb{R}}
\newcommand\Z{\mathbb{Z}}

\newcommand\AAA{\mathcal{A}}
\newcommand\III{\mathcal{I}}
\newcommand\KKK{\mathcal{K}}

\newcommand\OOO{\mathcal{O}}

\newcommand\SSS{\mathcal{S}}
\newcommand\TTT{\mathcal{T}}
\newcommand\VVV{\mathcal{V}}
\newcommand\WWW{\mathcal{W}}
\newcommand\CCC{\mathcal{C}}
\newcommand\BBB{\mathcal{B}}
\newcommand\XXX{\mathcal{X}}

\newcommand\Aff{\mathop{\mathrm{Aff}}}

\newcommand\End{\mathop{\mathrm{End}}\nolimits}
\newcommand\GL{\mathop{\mathrm{GL}}}

\newcommand\spann{\mathop{\mathrm{span}}\nolimits}

\def\id{\mathchoice{\setlength{\unitlength}{1ex}\identitaet}{
\setlength{\unitlength}{1ex}\identitaet}{
\setlength{\unitlength}{0.5ex}
\begin{picture}(1.6,1.5)
\put(0.7,-0.1){\scriptsize 1}\thinlines
\put(1,0.1){\line(0,1){1.4}}
\put(0.2,0){\line(1,0){1.2}}
\put(0.6,1.5){\line(1,0){0.4}}
\end{picture}}{\setlength{\unitlength}{0.5ex}
\begin{picture}(1.6,1.5)
\put(0.7,-0.1){\scriptsize 1}\thinlines
\put(1,0.1){\line(0,1){1.4}}
\put(0.2,0){\line(1,0){1.2}}
\put(0.6,1.5){\line(1,0){0.4}}
\end{picture}}
}
\newcommand\identitaet{\begin{picture}(1.6,1.5)
\put(0,0){1}
\put(1,0.1){\line(0,1){1.4}}
\thinlines
\put(0.2,0){\line(1,0){1.2}}
\put(0.6,1.5){\line(1,0){0.4}}
\end{picture}}

\newcommand\fa{\forall\,}

\newcommand\qu{\overline}

\newcommand\wt{\widetilde}
\newcommand\wh{\widehat}

\newtheorem{definition}{Definition}[section]

\newtheorem{prop}[definition]{Proposition}
\newtheorem{theorem}[definition]{Theorem}

\newenvironment{proof}{\textsl{Proof:}}{\hspace*{\fill}$\blacksquare$\\}

\captionnamefont{\bfseries}

\title{\mbox{}\hfill {\small DCPT-13/09}\\[20pt]
A twist in the $M_{24}$ moonshine story}

\author{\\ \Large{Anne Taormina\footnote{anne.taormina@durham.ac.uk}\;\; and Katrin Wendland\footnote{katrin.wendland@math.uni-freiburg.de}}\\ \\ \\
 \normalsize{$^*$Centre for Particle Theory, Department of Mathematical Sciences} \\ \normalsize{Durham University, Durham, DH1 3LE, U.K. }\\
 \normalsize{$^{\dagger}$Mathematics Institute, University of Freiburg}\\
 \normalsize{D-79104 Freiburg, Germany.  }}
 
 \date{\small{\textsc{AMS Subject Classification: 81T40, 81T60, 14J28}}}
 
\begin{document}
\maketitle
\setlength{\parindent}{0pt}
\begin{abstract}
Prompted by the Mathieu Moonshine observation,
we identify a pair of 45-dimensional vector spaces of states that account for the first order term in the 
massive sector of the elliptic genus of K3 in every $\Z_2$-orbifold CFT on K3.
These generic states are uniquely characterized by the fact that the action of 
every geometric symmetry group of a $\Z_2$-orbifold CFT 
yields a well-defined faithful representation on them.
Moreover, each such representation  
is  obtained by restriction of 
the $45$-dimensional irreducible representation
of the Mathieu group $M_{24}$ constructed by Margolin.
Thus we provide a piece of evidence for Mathieu Moonshine
explicitly from  SCFTs on K3.

The $45$-dimensional irreducible representation of $M_{24}$ 
exhibits a twist, which we prove can be undone 
in the case of $\Z_2$-orbifold CFTs on K3 for all geometric
symmetry groups. 
This twist however cannot be undone for the combined symmetry group $(\Z_2)^4\rtimes A_8$ 
that emerges from surfing the moduli space of Kummer K3s.
We conjecture that in general, the untwisted representations are
exclusively those of geometric symmetry groups in some 
geometric interpretation of a CFT on K3.
In that light, the twist appears as a representation theoretic manifestation of the maximality
constraints in Mukai's classification of geometric symmetry groups of K3. 
\end{abstract}
%%%%%%%%%%%%%%%%%%%%
\section*{Introduction}
%%%%%%%%%%%%%%%%%%%%
The Mathieu Moonshine observation \cite{eot10} continues to inspire three years on. 
It is now proven that the multiplicity spaces of irreducible characters of the $N=4$ superconformal algebra in
the elliptic genus of K3 
do indeed correspond to representations of the sporadic group 
$M_{24}$ \cite{ga12}. The reason why $M_{24}$ is singled out remains a mystery. From
the properties of twining elliptic genera, one may expect a representation of $M_{24}$ 
on a vertex algebra which governs the elliptic genus of K3, as  argued in 
\cite{gprv12,ga12,gpv13}. However, there are conceptual difficulties in following  
this lead, particularly in the sector of the elliptic genus corresponding to massless states 
at leading order.\\

In  a recent paper \cite{tawe12}, we suggest a starting point for the construction of a vertex 
algebra that governs the states occuring at lowest order in the elliptic genus. Our approach 
uses a subtle interplay between the geometry inherited from the K3 surfaces on which 
superstrings propagate and the (chiral, chiral) algebra associated with $N=(4,4)$ superconformal 
algebras at central charge $c=\qu c=6$ \cite{lvw89}. Since the elliptic genus is an invariant on the 
moduli space of such superconformal field theories (SCFTs), we are at liberty to choose the special 
class of $\Z_2$-orbifold  conformal field theories $\CCC=\TTT/\Z_2$ on K3 to present 
our arguments\footnote{Throughout our work, $\CCC=\TTT/\Z_2$ refers to the standard $\Z_2$-orbifold
construction induced by the Kummer construction for K3 surfaces.\label{whichz2}}, which we summarize here. 

The elliptic genus counts states (with signs) that 
appear in the Ramond-Ramond sector of the partition function, after projection onto 
Ramond ground states in the antiholomorphic sector of the theory. Although the 
expected {vertex} algebra $\wh\XXX$ cannot arise in the Ramond-Ramond sector of such SCFTs, 
one may, by choosing appropriate holomorphic and antiholomorphic $U(1)$-currents within the 
relevant $N=(4,4)$ superconformal algebra, spectral flow the states into the Neveu-Schwarz 
sector of the theory where (prior to all projections and truncations)
they yield a closed vertex algebra $\wh\XXX$. The Ramond-Ramond ground states, 
in particular, flow to (chiral, chiral) states.
In fact, the (chiral, chiral) algebra $\XXX$ of \cite{lvw89}, which accounts for the 
contributions to the lowest order terms of the elliptic genus, is obtained from
$\wh \XXX$ by truncation. In any theory $\CCC=\TTT/\Z_2$ on 
K3, the  vector space $\XXX$ is generated by the $24$ fields
\be\label{newbasis}
\xi_1\xi_2\xi_3\xi_4,\quad \xi_i\xi_j \;(1\leq i<j\leq 4),\quad \id;\quad \wt T_{\vec a}\; ({\vec a}\in\F_2^4),
\ee
where $\xi_i$ are holomorphic-antiholomorphic combinations of the Dirac fermions in the theory, 
$\id$ denotes the vacuum field, and the $\wt T_{\vec a},\, {\vec a}\in\F_2^4$, are the sixteen fields 
spectral-flowed from the RR twist fields.

Interestingly, the  truncation of the OPE to the (chiral, chiral) algebra $\XXX$ leaves the 
latter independent of all moduli. This may seem desirable, 
since the elliptic genus does not depend on the moduli. 
Mathieu Moonshine indicates that one 
should consider symmetries of  some underlying vertex algebra, while it 
is far from clear from inspecting the fields in \eqref{newbasis} {which linear maps are symmetries}
of the whole theory $\CCC=\TTT/\Z_2$.
In  \cite{tawe12}, we motivate why in our setting, we restrict our attention to 
symmetry groups that are induced 
geometrically in some geometric interpretation of $\CCC$
stemming from $\TTT$. In other words, all such symmetry groups  
are subgroups of $M_{24}$ by 
Mukai's seminal result \cite{mu88}. Imposing that the superconformal algebra of $\CCC$ be 
pointwise fixed also requires that a four-dimensional subspace of \eqref{newbasis} 
is fixed under symmetries. Therefore, this condition rules out the possibility for $\XXX$ 
to carry a representation of $M_{24}$, as also argued in \cite{gprv12} albeit from a different perspective.
Hence a vertex algebra which governs the leading order terms of the elliptic genus and which 
at the same time carries the expected representation of $M_{24}$ must be related to 
$\XXX$ by some nontrivial map. The Niemeier markings and the overarching maps 
which were constructed in \cite{tawe11} should be viewed as a first approach towards constructing such a map. \\

This indicates  interesting geometric avenues to explore while searching for a vertex 
algebra governing the leading order terms of the elliptic genus. {In this paper, 
we  take a closer look} at the leading order in the \textsl{massive} sector of the 
elliptic genus. After all, the Mathieu Moonshine observation \cite{eot10} originally 
refers to the massive sector. We use again the framework of $\Z_2$-orbifold  CFTs on K3, 
we analyze the massive states at leading order and show that they populate 
two complex $45$-dimensional representation spaces of the respective symmetry groups,
which we call $V_{45}^{CFT}$ and $\qu V_{45}^{CFT}$. This has long been anticipated from mere 
Moonshine numerology: the massive character with the lowest conformal weight appears 
with the coefficient $90$ in the elliptic genus, and $M_{24}$ has two complex conjugate irreducible 
$45$-dimensional representations. Nevertheless, it is remarkable that one obtains well-defined
representations of the symmetry groups from the \textsl{net contributions} to the elliptic genus in such a natural fashion.
We prove that the representations  on
$V_{45}^{CFT}$ and $\qu V_{45}^{CFT}$ can be induced from the two irreducible, complex conjugate $45$-dimensional representations of $M_{24}$.
By working within  specific orbifold theories 
$\CCC=\TTT/\Z_2$, we gain much deeper insights into the nature of these massive states, 
and crucially, appreciate how the symmetries of interest act on them. 
{In particular, we show that the space $V_{45}^{CFT}\oplus\qu V_{45}^{CFT}$ is
uniquely characterized by the quantum numbers of states yielding leading order massive
contributions to the elliptic genus together with the requirement that it carries a faithful 
representation of the \textsl{generic} geometric symmetry group $(\Z_2)^4$ of $\Z_2$-orbifold CFTs
$\CCC=\TTT/\Z_2$.}

An indispensable ingredient is  the work of Margolin \cite{ma93}, who 
constructs an irreducible $45$-dimensional representation of $M_{24}$ on the 
$45$-dimensional space $V_{45}$. There, an action of 
$(\Z_2)^4\rtimes A_8 \subset M_{24}$ on $V_{45}$ is exhibited as an important stepping stone 
in the construction of the full $M_{24}$ action. This maximal subgroup of $M_{24}$ is particularly 
relevant to us, as we have emphasized its role as combined symmetry group of all holomorphic 
symplectic automorphism groups of Kummer surfaces in \cite{tawe12}. 
There, we encrypt the
action of $(\Z_2)^4\rtimes A_8$
 unequivocally  on the Niemeier lattice with root lattice $A_1^{24}$,
which carries a natural representation of $M_{24}$. In that setting, the very 
fact that the Niemeier lattice has definite signature while the full integral homology lattice 
$H_\ast (X,\Z)$ of K3 surfaces $X$ has indefinite signature is a geometric obstruction 
against an action of $(\Z_2)^4\rtimes A_8$ and a fortiori of $M_{24}$ 
in terms of geometric symmetries on \textsl{any} 
K3 surface, a fact already appreciated by Mukai \cite{mu88} and Kondo \cite{ko98}. 

In the present work, we
discover a new potential  obstruction  from representation theory
to the action of $M_{24}$ as a geometric symmetry group. 
Indeed, we  prove that  the group action on $V_{45}$ simplifies when
one restricts from the group $(\Z_2)^4\rtimes A_8$  to the subgroups
which happen to be  symmetry groups of Kummer surfaces with dual K\"ahler
class induced from the underlying torus.  More precisely, we deduce from \cite{ma93} that
$V_{45}$ possesses the structure of a tensor product $V_{45}=\VVV\otimes \BBB$, where
the factorization however is not 
respected by
the action of $(\Z_2)^4\rtimes A_8$ on $V_{45}$. Here, $\VVV$ is a 
$3$-dimensional complex vector space, and the $15$-dimensional space $\BBB$
is referred to as the {\textsc{base} of $V_{45}$}. 
In fact,  $(\Z_2)^4\rtimes A_8$ permutes the $3$-dimensional ``fibers''  $\VVV\otimes\spann_\C\{B\},\,
B\in\BBB$,  
introducing a \textsc{twist}  on each fiber {(see Def.\ \ref{twist})}. 
Such a twist  is necessary for the construction of the 
(reducible) $45$-dimensional 
representation of that  group, since $(\Z_2)^4\rtimes A_8$
does not possess any nontrivial $3$-dimensional 
representations that $\VVV$ could carry. 
{The  above-mentioned simplification}  under restriction
to geometric symmetry groups
amounts to the fact that these groups act without a twist, as we prove. 
On the CFT side this is expected from
the very structure of the states in $V_{45}^{CFT}\oplus\qu  V_{45}^{CFT}$, which
forbids geometric symmetry group actions  induced from $\TTT$ to exhibit a twist. 

Nevertheless, in view of our search for an explanation to Mathieu Moonshine,
we  interpret the representation space $V_{45}$ as a medium which  combines the actions of symmetry
groups at distinct points of the moduli space of $N=(4,4)$ SCFTs on K3
to representations of larger groups, similarly to
the ideas we present in \cite{tawe11,tawe12}. 
Indeed, we generate
the action of the entire group $(\Z_2)^4\rtimes A_8$ on $V_{45}$ by combining
the actions of the maximal symmetry groups of Kummer K3s.
Moreover, we prove that this action of $(\Z_2)^4\rtimes A_8$ is obtained from 
Margolin's irreducible representation of $M_{24}$
by restriction to the maximal subgroup $(\Z_2)^4\rtimes A_8$. Thus we obtain a first piece of evidence
for an action of $M_{24}$ on a selection of states generic to all $\Z_2$-orbifold CFTs on K3.
 \\

We start in Section\,\ref{sec:counting} by a detailed account of the massive states described 
above. This 
exercise leads to Proposition\,\ref{our45}, which provides the mathematical structure 
organising $90$ twisted massive states into two $45$-dimensional 
spaces. More precisely, we find 
$V_{45}^{CFT}={\bf{3}}\otimes\bf{15}$ and $\qu V_{45}^{CFT}={\bf{\qu 3}}\otimes\bf{15}$, 
with ${\bf 3}$ and ${\bf \qu 3}$ being complex representation spaces of $SO(3)$, and ${\bf 15}$ being 
a representation space of $\Aff(\F_2^4)$ hosting twisted ground states. In Section\,\ref{sec:twisted}, we focus on the action of the 
maximal subgroup $(\Z_2)^4\rtimes A_8$ of $M_{24}$ on the base 
${\bf 15}$ of the space of states $V_{45}^{CFT}$.  
In Proposition\,\ref{repona} we show
that the representation of  $(\Z_2)^4\rtimes A_8$ on the space 
${\bf 15}$ of twisted ground states is equivalent to the representation of the same group
constructed by Margolin 
\cite{ma93}  on the $15$-dimensional base 
$\BBB$ of $V_{45}=\VVV\otimes\BBB$.  
Section\,\ref{sec:twist} analyzes the properties of the space $V_{45}^{CFT}$ in order to 
substantiate the expectation that for every $\Z_2$-orbifold CFT, this space carries a 
representation of a geometric symmetry group $G\subset M_{24}$ which is induced from
Margolin's representation of $M_{24}$ on $V_{45}$. A first step in proving 
this is to show that none of the symmetry 
groups of maximally symmetric Kummer surfaces acts with a twist in Margolin's
representation on $V_{45}$. This is the purpose 
of Propositions \ref{nosquaretwist}, \ref{notetratwist} and \ref{notriangletwist}. A second step is to 
prove that the representation of each of the three maximal symmetry groups of Kummer surfaces on 
$V_{45}^{CFT}$ is equivalent to the representation of that same group viewed as a subgroup of 
$(\Z_2)^4\rtimes A_8$ acting on  $V_{45}$. This is done in Proposition\,\ref{equivalence}. 
{Our main result,
Theorem \ref{result}, generalizes these findings} 
to arbitrary symmetry groups of $\Z_2$-orbifold CFTs which
are induced by geometric symmetries of the underlying toroidal theories.
We then briefly discuss the role and limitations of a  lattice of rank $20$ that accommodates the combined  
action of the three maximal symmetry groups of Kummer surfaces in our efforts to understand the role of 
$M_{24}$ on the CFT side. 
In the Appendix we collect the details of Margolin's construction that are relevant for our work.
%
%%%%%%%%%%%%%%%%%%%%%%%%%%%%%%%%%%%%
\section{Counting states in $\Z_2$-orbifold CFTs on K3}
\label{sec:counting}
%%%%%%%%%%%%%%%%%%%%%%%%%%%%%%%%%%%%
In this section, we use the conformal field theoretic elliptic genus of K3 to determine
a $45$-dimensional vector space $V_{45}^{CFT}$ of states, which 
exists in all $\Z_2$-orbifold conformal field
theories on K3 and which is expected to be related to a representation
of the Mathieu group $M_{24}$ by the Mathieu Moonshine phenomenon \cite{eot10}.

{We use the notion 
of CFTs on K3 which can be found, for example, in \cite{asmo94,nawe00}, along with 
many further relevant references to the topic. One may also consult the more recent 
publications \cite{we10, ghv11,we14}. In the present work, however, we solely address
$\mathbb Z_2$-orbifold CFTs $\mathcal C=\mathcal T/\mathbb Z_2$ with $\mathcal T$
a toroidal SCFT at central charge $c=6=\overline c$ and $\mathcal C$ the
standard $\mathbb Z_2$-orbifold of this theory, cf.\ footnote \ref{whichz2}. 
This ensures the mathematical foundations of the present work: Recall that a
definition of the underlying toroidal theories $\mathcal T$ has been given 
in \cite{ko03},  generalizing Kac's lattice algebras \cite{ka98} to non-holomorphic CFTs. 
For all cyclic groups $G$, thus including the case $G=\mathbb Z_2$ that is relevant to our
work, orbifold techniques for toroidal theories have been put on a solid mathematical foundation in a series of papers by Fr\"ohlich, 
Fuchs, Runkel and Schweigert, culminating in \cite{ffrs10}.
That $\mathcal C=\mathcal T/\mathbb Z_2$ as above obeys all defining properties 
of a SCFT on K3, in particular that it enjoys $N=(4,4)$ supersymmetry and that its
conformal field theoretic elliptic genus agrees with the geometric elliptic genus
of K3 surfaces, has been shown in \cite{eoty89}, see also \cite{asmo94,nawe00,diss,we14}.}
%
%%%%%%%%%%%%%%%%%%%%%%%%%%%%%%%%%%%%%%%%%%
\subsection{Evidence for the $\bf{45}\oplus\qu{\bf{45}}$ of $\mathbf{M_{24}}$ from the elliptic genus of  K3}\label{sec:evidence}
%%%%%%%%%%%%%%%%%%%%%%%%%%%%%%%%%%%%%%%%%%
The 
conformal field theoretic elliptic genus of  $N=(2,2)$ 
superconformal field theories (SCFTs) \cite{akmw87,wi87} yields an  invariant on 
every connected component of the moduli space of $N=(2,2)$ 
SCFTs. For SCFTs on K3, one has  central charges $c=\qu c= 6$ and 
$N=(4,4)$ supersymmetry, and 
the conformal field theoretic elliptic genus $Z_{K3}(\tau, z)$ may be defined as
$$
Z_{K3}(\tau,z):={\rm tr}_{{\cal H}^R}\!\!\left((-1)^F\,y^{J_0}\,q^{L_0-\frac{1}{4}}\,\qu q^{\qu L_0-\frac{1}{4}}\right),
\quad q:=e^{2\pi i\tau},\,y:=e^{2\pi i z},\;\tau,z\in\C, \Im(\tau)>0.
$$
Here, the trace is over the states belonging to the Ramond-Ramond sector ${\cal H}^R$
of the theory, 
$F$ is the fermion number operator, $J_0$ is the zero mode of a choice of $U(1)$-current 
in the holomorphic 
$N=4$ superconformal algebra, while $L_0$ (resp.\ $\qu L_0$) is the zero-mode of the 
holomorphic (resp.\ antiholomorphic) Virasoro field. 
In fact, the above definition implies that the K3 elliptic genus is obtained from
the partition function of \textsl{any} $N=(4,4)$ SCFT on K3 by
$$
Z_{K3}(\tau,z)=Z_{\wt{R}}(\tau,z;\qu\tau, \qu z=0),
$$
where $Z_{\wt{R}}(\tau,z;\qu\tau, \qu z)$ is the  Ramond-Ramond ($\wt R\wt R$) 
partition function with fermion number 
insertion in both the holomorphic and antiholomorphic sectors. 
By standard cohomological arguments, 
the insertion of $\qu z=0$ in the $\wt R\wt R$ partition function suppresses 
the dependence of the resulting function on $\qu\tau$. One therefore expects a decomposition of $Z_{K3}$ in 
terms of $N=4$ characters stemming from the holomorphic sector of the partition function. Such a 
decomposition was achieved in  \cite{eoty89}, where the $N=(4,4)$ SCFTs 
chosen for calculation were Gepner models at $c=\qu c=6$.  The very form in which the Mathieu 
Moonshine phenomenon was observed appears in  \cite{oo89,diss}. Indeed, in terms of 
characters of irreducible representations of the $N=4$ superconformal algebra one may write
\be\label{n4deco}
Z_{K3}(\tau,z)=-2\ch^{\wt{R}}_{1,\frac{1}{2}}(\tau,z)+20\ch^{\wt{R}}_{1,0}( \tau,z)
+e(\tau)\,\widetilde{\ch}^{\wt{R}}(\tau,z).
\ee
Here,  the $N=4$ massless characters $\ch^{\wt{R}}_{1, 0}(\tau,z)$ and 
$\ch^{\wt{R}}_{1, \hf}(\tau,z)$ may be obtained from the Ramond sector characters derived in 
\cite{egta88} through the {shift} $z \mapsto z+\hf$, and
with $h \in \mathbb{R},\, h >0$, the massive $N=4$ 
characters in this sector are of the form 
\be\label{massivechar}
q^h\,\wt{\ch}^{\wt{R}}(\tau,z)=q^{h-\frac{1}{8}}\,\frac{\vartheta^2_1(\tau,z)}{\eta^3(\tau)}
=q^h(2-y-y^{-1})\left(1+q(1-2y-2y^{-1})+\cdots\right).
\ee
The function $e(\tau)$ in \req{n4deco} is 
closely related to a weakly holomorphic mock modular form of weight $\hf$ on $SL(2,\mathbb{Z})$ \cite{dmz12}, and its 
$q$-expansion starts with
\be \label{e}
e(\tau)=90q+462q^2+\cdots .
\ee

The root of the Mathieu Moonshine phenomenon lies in the coefficients of the series \eqref{e}: 
it was observed in \cite{eot10} that these coefficients appeared to be twice the dimensions of some 
representations of the sporadic Mathieu group $M_{24}$. A proof of this fact,
along with its highly non-trivial generalizations to  twining genera,
was given recently in  \cite{ga12}, which builds on the works  \cite{ch10, ghv10a,ghv10b,eghi11}. 
The field theoretic reason for $M_{24}$ to act on the states in the massive sector that 
$Z_{K3}$ accounts for
has remained a mystery so far. In order to unveil some of this $M_{24}$ Moonshine Mystery, 
it seems natural to track the states of one's favourite $N=(4,4)$ SCFT on K3 and 
determine explicitly which ones contribute to
the elliptic genus in the form of representations of $M_{24}$ or its subgroups.
In this work we shall do so for the leading order term of \req{e} in the case
of any $\Z_2$-orbifold conformal field theory on K3, which we denote 
$\CCC=\TTT/\Z_2$, where $\TTT$ is the underlying toroidal CFT in four dimensions. 
The current section is devoted to determining a $45$-dimensional space $V_{45}^{CFT}$ of states
which is generic to all such theories, such that $V_{45}^{CFT}\oplus\qu V_{45}^{CFT}$
accounts for the leading order coefficient $90$ of $e(\tau)$ in \req{e}.
%
%%%%%%%%%%%%%%%%%%%%%%%%%%%%%%%%%%%%%%%%%%
\subsection{Counting massive states in $\Z_2$-orbifold conformal field theories}\label{sec:massivecounting}
%%%%%%%%%%%%%%%%%%%%%%%%%%%%%%%%%%%%%%%%%% 
Every toroidal conformal field theory $\TTT$
possesses two free Dirac fermions on the holomorphic side,
which we denote by $\chi^1_+(z),\, \chi^2_+(z)$. The 
{fields of the complex conjugates} are denoted 
$\chi^1_-(z),\, \chi^2_-(z)$, such that
\be\label{Diracfields}
\chi^k_+(z)\chi^\ell_-(w)\sim{\delta_{k\ell}\over z-w}, \qquad k,\,\ell\in\{1,\,2\},
\ee
while the antiholomorphic counterparts are denoted $\qu\chi^1_\pm(\qu z),\, 
\qu\chi^2_\pm(\qu z)$. The superpartners of the two Dirac fermions $\chi^1_+(z),\, \chi^2_+(z)$ are given by 
\be\label{complexcurrents}
\textstyle
j_+^{1}(z):=\frac{1}{\sqrt{2}}(j^{1}(z)+i\,j^{2}(z)) \quad\mbox{ and }\quad
j_+^{2}(z):=\frac{1}{\sqrt{2}}(j^{3}(z)+i\,j^{4}(z)),
\ee
with $j^K(z), K \in \{1,\ldots, 4\}$, 
four real holomorphic $U(1)$-currents. 
{We remark that the introduction of the fields $\chi_\pm^k(z),\, \qu\chi_\pm^k(\qu z)$ and their 
superpartners amounts to a choice of basis for fields with appropriate quantum numbers, which
is tantamount to a choice of geometric interpretation (see \cite{nawe00,diss} for extensive 
discussions of this issue). Indeed, the fields $j_\pm^k(z),\, k\in\{1,\,2\},$ are identified with the
holomorphic coordinate vector fields ${\partial\over\partial z_k},\, k\in\{1,\,2\},$
in such a geometric interpretation. As was argued in the introduction, in this work we are only
interested in geometric symmetries of our respective CFTs. For this notion to make sense, the choice
of a geometric interpretation is inevitable.}
As usual, in $\TTT$ we have the mode expansions
\be\label{modeexpansion}
j^K(z)  = \sum_{n \in \mathbb{Z}}a^K_nz^{n-1},\qquad
\chi^k_{\pm}(z) = \sum_{n\in\Z+r} (\chi_\pm^k)_n z^{n-1/2},
\ee
where $r=\hf$ in the Neveu-Schwarz sector and $r=0$ in the Ramond sector.
The charges with respect to $(j^1,\ldots,j^4;\qu \jmath^1,\ldots, \qu\jmath^4)$ are denoted 
$p:=(p_L;p_R)\in \Gamma \subset \R^{4,4}$, where $\Gamma$ is a self-dual even integral
lattice with signature $(4,4)$.\\

The $\mathbb{Z}_2$-orbifold action on the Dirac fermions is given by 
$\chi_\pm^k(z) \mapsto - \chi_\pm^k(z)$, and on the $U(1)$-currents
by $j^K(z) \mapsto -j^K(z)$. One may construct $\mathbb{Z}_2$-invariant  
generators of the $N=4$ superconformal algebra from these free fields, namely the $U(1)$-current
$$
J^3=\hf( \nop{\chi_+^1\chi_-^1} + \nop{\chi_+^2\chi_-^2} ),
$$
the energy-momentum tensor
$$
T=\nop{j_+^1j_-^1} + \nop{j_+^2j_-^2} +
\hf(\nop{\partial\chi_+^1\chi_-^1} + \nop{\partial\chi_-^1\chi_+^1} +
\nop{\partial\chi_+^2\chi_-^2} + \nop{\partial\chi_-^2\chi_+^2}),
$$
and the remaining $SU(2)$-currents and $N=4$ supercurrents
\be\label{sucha}
J^\pm=\pm\nop{\chi_\pm^1\chi_\pm^2},\qquad
G^\pm=\sqrt2( \nop{\chi_\pm^1 j_\mp^1} + \nop{\chi_\pm^2j_\mp^2} ),
\qquad
G^{\prime\pm}=\sqrt2( \nop{\chi_\mp^1 j_\mp^2} - \nop{\chi_\mp^2j_\mp^1} ).
\ee
Similar expressions are obtained in the antiholomorphic sector. 
We denote by $Q$ and $h$ (resp.\ $\qu Q$ and $\qu h$) the eigenvalues of 
$J^3_0$ and $L_0$ (resp.\ $\qu{J}^3_0$ and $\qu L_0$).\\

We label the four $RR$ ground states that correspond to the two Dirac fermions 
by
$\sigma_i^{\pm\pm},\, i \in\{1,2\}$. These states are odd under the $\Z_2$-action. 
We summarize the $RR$ ground state content of the untwisted sector of the $\Z_2$-orbifold conformal 
field theory
$\CCC=\TTT/\Z_2$ on K3  in Table\,\ref{tab:untwistedgs}.
\begin{table}[htbp]
  \centering
    \begin{tabular}{|l|l||l|l|}
    \hline
 \,\,\,\,\,charged   &&\,\,uncharged&\\
   ground state &$(h, Q;\qu h,\qu Q)$& ground state &$(h, Q;\qu h,\qu Q)$\\
   \hline
    &&&\\
    $\sigma_1^{++}\sigma_2^{++}$&$(\frac{1}{4},1;\frac{1}{4},1)$
    & $\sigma_1^{++}\sigma_2^{--}$&$(\frac{1}{4},0;\frac{1}{4},0)$\\[3pt]
 $\sigma_1^{+-}\sigma_2^{+-}$&$(\frac{1}{4},1;\frac{1}{4},-1)$
 &$\sigma_1^{--}\sigma_2^{++}$&$(\frac{1}{4},0;\frac{1}{4},0)$\\[3pt]
$\sigma_1^{-+}\sigma_2^{-+}$&$(\frac{1}{4},-1;\frac{1}{4},1)$
&$\sigma_1^{-+}\sigma_2^{+-}$&$(\frac{1}{4},0;\frac{1}{4},0)$\\[3pt]
$\sigma_1^{--}\sigma_2^{--}$&$(\frac{1}{4},-1;\frac{1}{4},-1)$
&$\sigma_1^{+-}\sigma_2^{-+}$&$(\frac{1}{4},0;\frac{1}{4},0)$\\
  &&&\\\hline
\end{tabular}
  \caption{\textsl{RR ground states in the untwisted sector of ${\cal T}/\Z_2$.}}
\label{tab:untwistedgs}
\end{table}
Note that each of the four charged $RR$ ground states 
can be obtained from, say, $\sigma:=\sigma_1^{--}\sigma_2^{--}$
by application of the zero-modes $J_0^+,\,\qu J_0^+$ of the $SU(2)$-currents listed 
in \req{sucha}; these states comprise the ground states of the vacuum representation
of the $N=(4,4)$ superconformal algebra in the Ramond-Ramond sector. On the other hand, each
of the uncharged RR ground states is the ground state of a massless matter representation.

Moreover, 
there is a $16$-dimensional space of twisted ground states 
in the Ramond-Ramond sector with orthonormal basis
$T_{\vec a}$, where each state has quantum numbers $(h, Q;\qu h,\qu Q)=(\frac{1}{4}, 0;\frac{1}{4},0)$. 
 In any geometric interpretation of $\CCC$ on a Kummer
surface with underlying torus $T=\R^4/\Lambda$, the label
$\vec a\in\F_2^4 \cong\hf\Lambda/\Lambda$ 
refers to the fixed point of $\Z_2$ at which the respective field is localized. \\

We now identify the states in our theory $\CCC={\cal T}/\mathbb{Z}_2$ which are expected 
to form a $\bf{45}\oplus\qu{\bf{45}}$ representation of $M_{24}$ 
because they generically contribute to the leading coefficient $90$ in \eqref{e}. 
In the next two sections, we show how a maximal subgroup $(\mathbb {Z}_2)^4\rtimes A_8$ of $M_{24}$ 
acts on these states, and we
establish a link with the geometric picture we have developed in \cite{tawe11,tawe12}.\\

The $\wt{R}\wt{R}$ partition function of  ${\cal T}/\mathbb{Z}_2$ may be read off \cite{eoty89, diss} after suitable spectral flow. We write
$$
Z^{\wt{R}}:=Z^{\wt{R}}_{{\rm untwisted}}+Z^{\wt{R}}_{{\rm twisted}},
$$
with 
\ba
Z^{\wt{R}}_{{\rm untwisted}}(\tau,z;\qu\tau, \qu z)
&=&\frac{1}{2|\eta(\tau)|^8}\,\left( 1+\smash{\sum_{\substack{(p_L;p_R)\in \Gamma, \\ (p_L;p_R) \neq (0;0)}}}\,\,
q^{\frac{p_L^2}{2}}\qu q^{\frac{p_R^2}{2}}\right) \bigg | \frac{\vartheta_1(\tau,z)}{\eta(\tau)}\bigg |^4
+8\,\bigg | \frac{\vartheta_2(\tau,z)}{\vartheta_2(\tau)}\bigg |^4\!\!, 
\vphantom{\sum_{\substack{\Gamma\\\Sigma_F}}} \label{untwist}\\
Z^{\wt{R}}_{{\rm twisted}}(\tau,z;\qu\tau, \qu z)
&=&8\,\bigg | \frac{\vartheta_3(\tau,z)}{\vartheta_3(\tau)}\bigg |^4
+8\,\bigg | \frac{\vartheta_4(\tau,z)}{\vartheta_4(\tau)}\bigg |^4, \label{twisted}
\ea
with $\Gamma$ the charge lattice. 
From \req{n4deco} and \req{massivechar} we deduce that the massive states which contribute to the
leading order term of $e(\tau)$ have weights $(h;\qu h)=(\frac{5}{4}; \frac{1}{4})$. Moreover, it suffices
to focus on
 the states with quantum numbers $(h,Q;\qu h, \qu Q)=(\frac{5}{4}, 1; \frac{1}{4}, \qu Q)$.
%%%%%%%%%%%%%%%%%%%%%%%%%%%%%%%%%%%%%%%%%%
\subsubsection*{The untwisted sector}
%%%%%%%%%%%%%%%%%%%%%%%%%%%%%%%%%%%%%%%%%%
We first identify all states with these quantum numbers which generically
come from the untwisted sector \eqref{untwist}. Since each factor $1/\eta(\tau)$ 
counts states created from the vacuum by the bosonic oscillators 
$a_n^K,\, n\in \N,$ for fixed $K$, and since each factor 
$\vartheta_1(\tau,z)/\eta(\tau)$ counts bosonic and fermionic states (with signs) 
created from the vacuum by the fermionic
modes $(\chi^k_\pm)_n,\,n\in \N$, for a fixed value of $k$, one can read  from the expansion
$$% \label{theta1}
\frac{1}{\eta^4(\tau)}
\left( \frac{\vartheta_1(\tau,z)}{\eta(\tau)}\right)^2=-y^{-1}(1-2y+y^2)(1+q\,(4-2y-2y^{-1})+
\cdots)
$$
that the factor $y^{-1}$ accounts for a Ramond ground state with 
$U(1)$-charge $Q=-1$, the term $(-2y)$ in $(1-2y+y^2)$ accounts for the two fermionic zero modes 
$(\chi^k_+)_0,\,k \in\{1,2\},$  while the term $y^2$ accounts for the bilinear $(\chi_+^1)_0(\chi_+^2)_0$. 
On the other hand, the term $4q$ in the factor $(1+q\,(4-2y-2y^{-1})+\cdots)$ accounts for the four 
bosonic oscillators $a_1^K$ , while $(-2yq)$ and $(-2y^{-1}q)$ account for $(\chi^k_+)_1$ and 
$(\chi^k_-)_1,\,k\in\{1,2\}$, respectively. 
The states with  quantum numbers $(\frac{5}{4}, 1; \frac{1}{4}, \qu Q)$ 
in  $\displaystyle{\frac{1}{|\eta(\tau)|^8}\bigg | \frac{\vartheta_1(\tau,z)}{\eta(\tau)}\bigg |^4}$ 
that are $\Z_2$-invariant 
are thus encoded in the terms
\be \label{theta1term}
(y^{-1}\qu y^{-1})\left( (y^2)\,(4q)\,(-2\qu y)\, + (-2y)\,(-2qy)\,(1+\qu y^2)\right).
%+(y^{-1}\qu y^{-1})\,(4q)\,(y^2)\,(1+\qu y^2)+(y^{-1}\qu y^{-1})\,(-2qy)\,(-2y)(-2\qu y),
\ee
Since the charge lattice $\Gamma$ depends on the moduli of $\TTT$, and since the term
$\displaystyle{\bigg | 2\frac{\vartheta_2(\tau,z)}{\vartheta_2(\tau)}\bigg |^4}$ 
in \req{twisted} implements  the projection onto the $\Z_2$-invariant states, \req{theta1term}
accounts for all those $\Z_2$-invariant untwisted states with quantum numbers $(\frac{5}{4}, 1; \frac{1}{4}, \qu Q)$
which exist in every $\Z_2$-orbifold conformal field theory $\CCC=\TTT/\Z_2$.
Hence,
the generic contribution from the untwisted sector of the theory to this class of states 
amounts to 
eight fermions and eight bosons. Indeed, the eight fermions are  (with signs) 
given by\footnote{From now on, we set $\chi^\ell:=\chi_+^\ell$ and $\qu\chi^\ell:=\qu{\chi}_+^\ell,\,\,\ell\in \{1, 2\}$, 
as the fields $\chi_-^\ell$ and $\qu{\chi}_-^\ell$ do not appear in the expressions of the states we are considering.}
\be \label{fermions}
(4q)\,(y^2)\,(-2\qu y)\,(y^{-1}\qu y ^{-1})\colon\qquad 
a_1^K\,\chi_0^1\,\chi_0^2\,\qu \chi_0^{\ell}\,\sigma,\quad K \in\{1,\ldots,4\},\;  \ell\in\{1,2\},\quad[\qu Q=0],
\ee
and the eight bosons are given by
\be \label{bosons}
\begin{array}{lrl}
(2qy)\,(2y)\,(1)\,(y^{-1}\qu y^{-1})\colon&
\chi_{1}^k\,\chi_0^\ell\,\sigma& [\qu Q=-1],\\[5pt]
(2qy)\,(2y)\,(\qu y^2)\,(y^{-1}\qu y^{-1})\colon&
\chi_1^k\,\chi_0^\ell\,\qu{\chi}_0^1\qu{\chi}_0^2\,\sigma&   [\qu Q=1],
\end{array}
\quad
k,\ell \in\{1,2\},
\ee
where $\sigma:=\sigma_1^{--}\sigma_2^{--}$ is the ground state with charges 
$(Q;\qu Q)=(-1;-1)$ from Table\,\ref{tab:untwistedgs},
such that the $\chi_0^k\,\qu\chi_0^\ell\,\sigma$ with $k,\,\ell\in\{1,\,2\}$ yield the four
uncharged states from Table\,\ref{tab:untwistedgs}, 
while $\qu\chi_0^1\,\qu\chi_0^2\,\sigma$ yields the one with charges $(Q;\qu Q)=(-1;1)$.
Actually, the eight fermionic states in \eqref{fermions} are massless, 
since they are the images of the massless matter states $\chi_0^k\,\qu\chi_0^\ell\,\sigma$, 
$k,\ell\in\{1,\,2\}$, under the modes $G_1^+,\,G^{\prime-}_1$ of the $N=4$ supercurrents
listed in \req{sucha}, respectively.
Each set of 
four bosons in \eqref{bosons} consists of one massless boson 
$L_1J_0^+ \sigma$ or $L_1J_0^+\qu\chi_0^1\,\qu\chi_0^2\,\sigma$,
and three massive ones. 
The occurrence of six massive contributions in total can 
also readily be checked by rewriting the $\wt R\wt R$ partition function of 
the untwisted sector of $\CCC=\TTT/\Z_2$ in terms of $N=4$ characters. Indeed,
using$$
\begin{array}{rcll}
\wt{\ch}^{\wt{R}}_{1, \textstyle\hf}( \tau,z)&=&-(y+y^{-1})+q(2-y-y^{-1})+\cdots,&\mbox{vacuum, massless}\\[3pt]
\wt{\ch}^{\wt{R}}_{1, 0}(\tau,z)&=&1+q(2-2y-2y^{-1}+y^2+y^{-2})+\cdots,&\mbox{massless matter}\\[3pt]
\wt{\ch}^{\wt{R}}(\tau,z)&=&2-y-y^{-1}+q(6-5y-5y^{-1}+2y^2+2y^{-2})+\cdots,&\mbox{massive}
\end{array}
$$
one obtains
$$%\label{charchar}
\hspace*{-1em}
\begin{array}{rcl}
Z^{\wt{R}}_{{\rm untwisted}}(\tau,z;\qu\tau, \qu z)
=\wt{\ch}^{\wt{R}}_{1, \textstyle\hf}( \tau,z)
{\qu{\wt{\ch}^{\wt{R}}_{1,  \textstyle\hf}}(\qu\tau, \qu z)}
&+&4\, \wt\ch^{\wt{R}}_{1, 0}( \tau,z)\qu{ \wt\ch^{\wt{R}}_{1, 0}}( \qu\tau, \qu z)\\[5pt]
&+&3q\,\wt{\ch}^{\wt{R}}(\tau,z)\qu{\wt\ch^{\wt{R}}_{1, \textstyle\hf}}(\qu\tau, \qu z)+\cdots,
\end{array}
$$
which  reflects accurately the order $q$ (and $\qu q^0$) contribution. 
In summary, the $16$ generic Ramond-Ramond states with quantum numbers 
$(\frac{5}{4}, 1; \frac{1}{4}, \qu Q)$ in the untwisted sector of $\CCC=\TTT/\Z_2$ are
\begin{itemize}
\item one massless boson with $\qu Q=1$ and one massless boson with $\qu Q=-1$,
\item eight massless fermions with $\qu Q=0$,
\item three massive bosons with $\qu Q=1$ and three massive bosons with $\qu Q=-1$.
\end{itemize}
These massive bosons contribute a term $-6q$ to the function $e(\tau)$ in \eqref{e}. \\
%%%%%%%%%%%%%%%%%%%%%%%%%%%%%%%%%%%%%%%%%%
\subsubsection*{The twisted sector}
%%%%%%%%%%%%%%%%%%%%%%%%%%%%%%%%%%%%%%%%%%
On the other hand, the 
expansion of $Z^{\wt{R}}_{{\rm twisted}}(\tau,z;\qu\tau, \qu z)$ encodes the contribution to the states with 
quantum numbers $(\frac{5}{4}, 1; \frac{1}{4}, \qu Q)$ in the term $16(-8qy)$, where the factor $16$  
accounts for the number of twisted ground states
$T_{\vec {a}},\, \vec{a}\in \F_2^4$. By a similar analysis as above,
the $128$ states of interest are fermionic 
of the form\footnote{Recall that the mode expansion \req{modeexpansion} of the
$U(1)$-currents and free fermions
in the twisted RR sector has modes $a_n^K,\,\chi_n^\ell$ with $n\in\Z+\hf$.}
$a_\hf^K\,\chi_\hf^\ell\,T_{\vec{a}}$, $K \in \{1,\ldots,4\},\, \ell\in\{1,2\}$, with $\qu Q=0$. 
Since one can write
$$
Z^{\wt{R}}_{{\rm twisted}}(\tau,z;\qu\tau, \qu z)
=16\,\left \{\wt\ch^{\wt{R}}_{1, 0}( \tau,z)\qu{ \wt\ch^{\wt{R}}_{1, 0}}( \qu\tau, \qu z)
+6q\,\wt\ch^{\wt{R}}(\tau,z)\qu{\wt\ch^{\wt{R}}_{1, 0}}( \qu\tau, \qu z)+\cdots \right \},
$$
it follows that $32$ of these states are massless, while the remaining $96$ are massive, 
and contribute $+96q$ to \eqref{e}. 
Since our goal is to study these massive
states, it is imperative to 
determine which combinations of the $a_\hf^K\,\chi_\hf^\ell\,T_{\vec{a}}$ are massless. 
These are the $32$ states 
which are created from the
massless ground states $T_{\vec a},\,\vec a\in\F_2^4$, by 
the modes $G_1^+,\,G^{\prime-}_1$ of the $N=4$ supercurrents
listed in \req{sucha}.
We therefore 
introduce the modes of the complex currents \req{complexcurrents},
\begin{eqnarray*}
(j^1_+)_n&:=&\frac{1}{\sqrt{2}}(a^1_n+ia^2_n),\qquad j_-^1:=(j_+^1)^\ast,\nonumber\\
(j^2_+)_n&:=&\frac{1}{\sqrt{2}}(a^3_n+ia^4_n),\qquad j_-^2:=(j_+^2)^\ast,
\end{eqnarray*}
and from \req{sucha} we
find that the $32$ massless states
 may be written as\footnote{We suppress the modes for ease of reading.}
\be\label{masslesstwo}
\left(\chi^1\,j_-^1+\chi^2\,j_-^2\right)T_{\vec a}, \qquad 
\left(\chi^1\,j_+^2-\chi^2\,j_+^1\right)T_{\vec a},\quad \vec a\in\F_2^4.
\ee
The massive states  are perpendicular to these massless ones with respect to the standard
metric induced by the Zamolodchikov metric. Therefore we conveniently set
\be\label{three}
3:=\{\chi^1\,j_+^2+\chi^2\,j^1_+,\,\,\chi^1\,j^1_+,\,\,\chi^2\,j^2_+\,\},\qquad
\qu 3:=\{\chi^1\,j_-^1-\chi^2\,j^2_-,\,\,\chi^1\,j^2_-,\,\,\chi^2\,j^1_-\,\},
\ee
such that the $96$-dimensional vector space of massive twisted states with quantum numbers 
$(\frac{5}{4}, 1; \frac{1}{4}, \qu Q)$ has the basis
\be\label{pre45}
\left\{W T_{\vec a}\mid W\in 3\cup\qu 3,\;\vec a\in\F_2^4\right\}.
\ee
Note that all these states are fermionic.
%
%%%%%%%%%%%%%%%%%%%%%%%%%%%%%%%%%%%%%%%%%%
\subsection{A generic space $V_{45}^{CFT}$  in $\Z_2$-orbifold CFTs on K3} \label{sec:generic45}
%%%%%%%%%%%%%%%%%%%%%%%%%%%%%%%%%%%%%%%%%%
In the previous subsection, we have determined all massive states with quantum numbers 
$(\frac{5}{4}, 1; \frac{1}{4}, \qu Q)$ which exist generically in $\Z_2$-orbifold conformal field theories
on K3. Indeed, we have recovered  a $6$-dimensional space of untwisted bosonic states, along
with a $96$-dimensional space of twisted fermionic states, correctly accounting for a net contribution
of $-90qy$ to the elliptic genus. 
As explained in Section\,\ref{sec:evidence}, our interpretation of the Mathieu Moonshine
observations predicts that a pair $V_{45}^{CFT}\oplus\qu V_{45}^{CFT}$
of $45$-dimensional representation spaces of the Mathieu group $M_{24}$ should arise
from these states. While the emergence of the group $M_{24}$ remains mysterious,
we  expect to observe, on the space $V_{45}^{CFT}$, the representation of  subgroups of $M_{24}$ 
which occur as geometric symmetry groups of $\Z_2$-orbifold limits of K3 surfaces, 
in accordance with ideas already promoted in our previous works \cite{tawe11,tawe12}.

Indeed, we focus on  symmetry groups of SCFTs $\CCC=\TTT/\Z_2$  which
are induced by geometric symmetries of the underlying toroidal conformal field theories.\footnote{This 
includes the symmetries which are induced by
shifts by half lattice vectors on the underlying toroidal theory.}
{We emphasize that this notion only makes sense \textsl{after} the \textsl{choice}
of a geometric interpretation for the theory $\mathcal T$ on some torus $\mathbb R^4/\Lambda$.
As is explained in detail in \cite{tawe11,tawe12}, we even have to make a choice
of generators for the lattice $\Lambda$, and this means that in fact we are working
on a cover of the moduli space of SCFTs on K3. These choices in particular induce an
identification ${1\over2}\Lambda/\Lambda\cong\mathbb F_2^4$, such that every geometric
symmetry group $G$} acts on the 
twisted ground states $T_{\vec a},\,\vec a\in\F_2^4$, as permutation group
by means of affine linear maps on the space of labels  $\F_2^4$. In other words, we have
a natural representation
$$
R_G\colon G\longrightarrow \Aff(\F_2^4), 
$$
{once the very choices listed above have been made;}
see
\cite{tawe12} and Section\,\ref{sec:recap} for details. Furthermore, $G$ acts 
linearly as subgroup of $SO(3)$
on the $U(1)$-currents  $j^1_+,\,j^2_+$ of \req{complexcurrents}. More precisely,
$j^1_+,\,j^2_+$ form  a doublet $\bf{2}$ of $SU(2)$, as do
their fermionic superpartners $\chi^1_+,\,\chi^2_+$ of \req{Diracfields},
while $j^1_-,\,j^2_-$ carry a $\qu{\bf{2}}$.
 By a direct calculation one checks that the states \req{masslesstwo}
are invariant under the resulting action of $SU(2)$ and that the action respects the decomposition
\req{three}. In fact, we have $\bf{2}\otimes\bf{2}=\bf{1}\oplus\bf{3}$, 
$\bf{2}\otimes\qu{\bf{2}}=\bf{1}\oplus\qu{\bf{3}}$.
Since $-\id\in SU(2)$ acts trivially on $\bf{2}\otimes\qu{\bf{2}}$, we have an action of
$SO(3)=SU(2)/\{\pm\id\}$ on $\bf{3}$ and $\bf{\qu 3}$.
This action thus descends to a representation on 
the $96$-dimensional space of massive twisted states
with basis \req{pre45}, as does the action of $\Aff(\F_2^4)$ on the indices of the twisted
ground states. We formally denote the resulting representation of $SO(3)$ by $S$, where
\be \label{genericstates}
\wh V:= \spann_\C\left\{ W T_{\vec a} \mid W\in 3\cup\qu{3},\; \vec a\in\F_2^4\right\},\qquad
S\colon SO(3)\longrightarrow\End_\C(\wh V).
\ee
Every symmetry group $G$ of  $\CCC=\TTT/\Z_2$ which is induced by geometric symmetries of $\TTT$ 
has the form $G=(\Z_2)^4\rtimes G_T$ with $G_T\subset SO(3)$, see for example
\cite{tawe11,tawe12} for an exposition. Then the representation 
 $R^{CFT}_G\colon G\longrightarrow\End_\C(\wh V)$ obtained from the symmetries of $\CCC$ is given by 
\ba\label{actiondetails}
\fa g=(\vec c,g_T)\in G=(\Z_2)^4\rtimes G_T, \; W\in 3\cup\qu{3},\, \vec a\in\F_2^4\colon\nonumber\\
R^{CFT}_G(g)\left( WT_{\vec a}\right) &=& S(g_T)(W) T_{R_G(g)(\vec a)}.
\ea
It is important to note that 
$\wh V$ thereby is simply a tensor product of  the representation spaces $\bf{3}\oplus\qu{\bf{3}}$ of $SO(3)$ by 
a $16$-dimensional representation space of $\Aff(\F_2^4)$,
a fact that will be crucial later on, when we discuss group actions on this space of states. 

In fact, we immediately obtain a natural decomposition of $\wh V$ according to
$$
\bf{96} =(\bf{3}\oplus\qu{\bf{3}})\otimes\bf{1}\oplus(\bf{3}\oplus\qu{\bf{3}})\otimes\bf{15}
$$
as follows:
we decompose the $16$-dimensional space of twisted ground states 
into a one-dimensional space generated by 
$N_{0000}:=\frac{1}{4}\sum_{\vec{a}\in \F_2^4}^{16}T_{\vec{a}}$, and its 
orthogonal complement $\AAA$. 
Since $N_{0000}$ is invariant under the action of $\Aff(\F_2^4)$, this action descends to a representation
on $\AAA$. We now obtain the desired $45$-dimensional vector space
$V_{45}^{CFT}$ as the space which is generated by the states 
$W A$ with $W\in  \{\chi^1\,j_+^2+\chi^2\,j^1_+,\,\,\chi^1\,j^1_+,\,\,\chi^2\,j^2_+\,\}$ and
$A\in\AAA$. Then
$\qu V_{45}^{CFT}$ is defined analogously by using the ${\qu 3}$ from \req{three} instead of the ${3}$
as above. 
The restrictions of the representations $R_G$ and $S$ to $V_{45}^{CFT}$ and
$\qu V_{45}^{CFT}$ are denoted by $R_G$ and $S$ as well.
While at this point the choice of $V_{45}^{CFT}\oplus\qu V_{45}^{CFT}$ in $\wh V$ is only
justified by the fact that it is natural and compatible with a restriction of the representations 
$R_G$ and $S$, in Theorem \ref{result} we
prove that these spaces are in fact uniquely determined.\\

In summary, we have obtained the result that the generic field content of $\Z_2$-orbifold
conformal field theories on K3 ensures the existence of a space of states which naturally accounts for the massive net
contributions to the elliptic genus in leading order:
\begin{prop}\label{our45}
Consider the orthogonal complement $\AAA$ of $N_{0000}:=\frac{1}{4}\sum_{\vec{a}\in \F_2^4}^{16}T_{\vec{a}}$
in the space of twisted ground states of an arbitrary $\Z_2$-orbifold conformal field theory on K3. 
Then the space  
$$
V_{45}^{CFT} := \spann_\C\left\{ W A\mid W\in  \{\chi^1\,j_+^2+\chi^2\,j^1_+,\,\,\chi^1\,j^1_+,\,\,\chi^2\,j^2_+\,\}
,\;A\in\AAA\right\}
$$
is a $45$-dimensional vector space of massive states which together with $\qu V_{45}^{CFT}$
accounts for the leading order contribution to the function $e(\tau)$ that governs the elliptic genus
of K3 according to \mbox{\rm\req{n4deco}}. In terms of representations of symmetry groups,
$V_{45}^{CFT}$ is a tensor product $\WWW\otimes\AAA$, where  $\WWW$ is the three-dimensional
representation space $\bf{3}$ of $SO(3)$,
while $\AAA$ is a $15$-dimensional representation space of $\Aff(\F_2^4)$.
Similarly, $\qu V_{45}^{CFT}=\qu\WWW\otimes\AAA=\qu{\bf{3}}\otimes\bf{15}$.
\end{prop}
From the above proposition it follows that the vector space underlying $V_{45}^{CFT}$ serves as a medium to collect 
the actions of the geometric symmetry groups when 
symmetry-surfing the moduli space of $\Z_2$-orbifold conformal
field theories $\CCC=\TTT/\Z_2$ on K3. In the remaining sections of this paper, we show that
the combined action of these symmetry groups generates an action of $\Aff(\F_2^4)$ which
can be induced from a $45$-dimensional
irreducible representation of $M_{24}$ by restriction to $\Aff(\F_2^4){\cong(\mathbb Z_2)^4\rtimes A_8}$. 

To clear notations, we make use of the fact that $N_{0000}$ is invariant under both
the action of $SO(3)$ and of $\Aff(\F_2^4)$. We
let 
$$
\WWW_0:=\left\{W N_{0000}\mid W\in\{\chi^1\,j_+^2+\chi^2\,j^1_+,\,\,\chi^1\,j^1_+,\,\,\chi^2\,j^2_+\,\}\right\}\cong\WWW
$$ 
and frequently view $S$ as
a representation on $\WWW_0$,
$$
S\colon SO(3) \longrightarrow \End_\C(\WWW_0),
$$
the \textsc{model fiber} of $V_{45}^{CFT}$.
 %
%%%%%%%%%%%%%%%%%%%%%%%%%%%%%%%%%%%%%%%%
\section{The action of $\mathbf{(\Z_2)^4\rtimes A_8}$ on twisted ground states}
\label{sec:twisted}
%%%%%%%%%%%%%%%%%%%%%%%%%%%%%%%%%%%%%%%%%
In  Prop.~\ref{our45}, we have determined a $45$-dimensional space $V_{45}^{CFT}$ of 
states which in every theory $\CCC=\TTT/\Z_2$ 
 yields a representation
of the group of  symmetries  induced from geometric symmetries of the
underlying toroidal theory. 
{This representation depends on a choice of geometric interpretation for the theory $\mathcal T$
on some torus $\mathbb R^4/\Lambda$ together with a choice of generators for $\Lambda$,
thus lifting our construction onto a cover of the moduli space of SCFTs.}
As explained in Section\,\ref{sec:evidence},
Mathieu Moonshine predicts that this space is in fact related to a representation
of the Mathieu group $M_{24}$. Indeed, $M_{24}$ possesses an irreducible representation of
dimension $45$ on the space $V_{45}$ which has been constructed by Margolin \cite{ma93}, 
see Appendix\,\ref{sec:margolin} for a summary. Eventually
we would like to understand how our space of states $V_{45}^{CFT}$ can be identified with the space $V_{45}$ 
as a representation of $M_{24}$. 
In the present work we focus on the action of the maximal subgroup 
$(\Z_2)^4\rtimes A_8{\cong\Aff(\mathbb F_2^4)}$ of $M_{24}$, which is obtained by combining all
groups of symmetries of SCFTs $\CCC=\TTT/\Z_2$  induced by geometric symmetry
groups of the underlying toroidal theory.\\

The structure of the space $V_{45}^{CFT}$ of fields obtained in the previous section is that of a
tensor product $\WWW\otimes\AAA$, where $\AAA$ is a fifteen-dimensional space of
 twisted ground states in our $\Z_2$-orbifold CFT, and $\WWW$ is three-dimensional
 and furnishes a triplet $\bf{3}$ of $SO(3)$. 
By choice of an appropriate orthonormal basis $\left\{N_X\mid X\in\{A,\,B,\,\ldots,\,N,\,O\}\right\}$ of $\AAA$
one can thus write this space in a form which is very similar to the form of $V_{45}$ given in
\req{45d},
$$
V_{45}^{CFT}=\WWW_A\oplus\WWW_B\oplus\cdots\WWW_N\oplus\WWW_O,\quad
\WWW_X:=\spann_{\C}\{ N_X\}\otimes\WWW\quad\fa X \in\{A,\,B,\,\ldots,\,N,\,O\}.
$$
By construction, see Appendix\,\ref{sec:margolin}, Margolin's representation $M$ of $(\Z_2)^4\rtimes A_8$
induces a well-defined action on the fifteen-dimensional vector space 
which we call the base $\BBB$ of $V_{45}$, and which is generated by the counterparts 
$P_X$ (see the discussion of \req{45d}) of the `CFT' orthonormal
basis $\left\{N_X\mid X\in\{A,\,\ldots,\,O\}\right\}$ of $\AAA$.
On the other hand, in Section\,\ref{sec:generic45} we mentioned that $(\Z_2)^4\rtimes A_8\cong\Aff(\F_2^4)$ acts
naturally on $\AAA$ by affine linear maps on the indices of the twisted ground states $T_{\vec a},\,\vec a\in\F_2^4$. Hence
 we expect that the latter representation of $(\Z_2)^4\rtimes A_8$ 
on $\AAA$
 is equivalent to the representation $M$ of this group on $\BBB$
described in Appendix\,\ref{sec:margolin}. In the current section, we
prove that this expectation holds true.

To make the claim precise, let us describe the space $\AAA$ in more detail.
The space of twisted ground states in our CFT $\CCC=\TTT/\Z_2$ 
has a natural orthonormal basis
$\left\{T_{\vec a} \mid \vec a\in\F_2^4\right\}$, where  $\vec a\in\F_2^4$ labels
the sixteen resolved singular points 
in any geometric interpretation on an orbifold limit of K3. As is explained
in our previous work \cite{tawe11,tawe12}, the group $(\Z_2)^4\rtimes A_8\cong\Aff(\F_2^4)$
therefore acts naturally on these states by affine linear transformations on the indices $\vec a\in\F_2^4$.
In Prop.~\ref{our45} the space $\AAA$ is obtained as the orthogonal
complement of the state $N_{0000}:={1\over4}\sum_{\vec a\in\F_2^4} T_{\vec a}$, which
is invariant under the action of $(\Z_2)^4\rtimes A_8\cong\Aff(\F_2^4)$ by construction.
The space $\AAA$ hence indeed carries an action  of $(\Z_2)^4\rtimes A_8\cong\Aff(\F_2^4)$. 
\begin{prop}\label{repona}
Consider the orthogonal complement
$\AAA$ of the state $N_{0000}={1\over4}\sum_{\vec a\in\F_2^4} T_{\vec a}$ 
in the space of twisted ground states in a $\Z_2$-orbifold conformal field
theory on K3. The natural representation of $(\Z_2)^4\rtimes A_8\cong\Aff(\F_2^4)$  on $\AAA$ 
through affine linear transformations of the indices $\vec a\in\F_2^4$ of the twisted ground
states $T_{\vec a}$
is equivalent to the representation $M$ of 
$(\Z_2)^4\rtimes A_8$ constructed by Margolin  on the base $\BBB$ of $V_{45}$.
\end{prop}
We postpone the proof of  Prop.~\ref{repona} to Section\,\ref{sec:proofofprop}, since as a preparation 
and for later convenience
we first recall some of the constructions and notations of \cite{tawe11,tawe12}. 
%
%%%%%%%%%%%%%%
\subsection{The action of $\mathbf{(\Z_2)^4\rtimes A_8}$ as combined symmetry group}\label{sec:recap}
%%%%%%%%%%%%%%
In \cite{tawe11,tawe12} we show that the group $(\Z_2)^4\rtimes A_8\cong\Aff(\F_2^4)$ 
can be obtained by combining the symmetry groups
of the three maximally symmetric Kummer surfaces\footnote{On a K3 surface $X$,
we call a biholomorphic map $f\colon X\longrightarrow X$ a \textsc{symmetry}, if its induced action on cohomology 
fixes the holomorphic volume form and 
the dual K\"ahler class of $X$.
As is explained  in \cite{tawe11,tawe12}, this implies that all symmetry groups of K3 surfaces are finite.
We require all our  Kummer surfaces to
be equipped with the dual K\"ahler class which is induced from the standard Euclidean
metric on the underlying torus. }. 
More precisely, the images $R_{G_k}(G_k)$, $k\in\{0,\,1,\,2\}$, of these three groups under their
natural representations on $\F_2^4$ generate the entire group $\Aff(\F_2^4)$.
The three maximally symmetric Kummer surfaces are the square Kummer surface $X_0$ with
symmetry group $G_0:=(\Z_2)^4\rtimes (\Z_2\times\Z_2)$, the tetrahedral Kummer
surface $X_1$ with symmetry group $G_1:=(\Z_2)^4\rtimes A_4$,  and the triangular Kummer surface $X_2$
with symmetry group $G_2:=(\Z_2)^4\rtimes S_3$. 
Let us denote by $\Lambda_k$, $k\in\{0,\, 1,\,2\}$, the defining
lattices of the complex tori underlying the Kummer surfaces $X_k$. Then
each  group $G_k$  acts on the twisted ground states 
$T_{\vec a}$ through  the permutations induced on $\F_2^4\cong{1\over2}\Lambda_k/\Lambda_k$ 
by the geometric action on  $\Lambda_k$. This defines the representations $R_{G_k}\colon G_k\longrightarrow\Aff(\F_2^4)$.

Let us fix some additional notations. The translational subgroup $(\Z_2)^4$ is common to all 
symmetry groups of Kummer K3s, and its elements
$\iota_{\vec c}$ with $\vec c\in \F_2^4$ act by
\be\label{translation}
\iota_{\vec c}\colon \quad T_{\vec a}\longmapsto T_{\vec a+\vec c}\;\;\fa \vec a\in\F_2^4
\ee
on the twisted ground states. 
To realize the action of the non-translational part of each symmetry group $G_k$, 
we first fix convenient generators for each of the lattices $\Lambda_k$ and for the groups $G_k$, $k\in\{0,\, 1,\,2\}$, 
see \cite[(1.5)-(1.9)]{tawe12} for our particular choices. 
In the case of the square Kummer surface $X_0$ with
non-translational symmetry group $\Z_2\times\Z_2$,  we introduce two generators 
$\alpha_1,\,\alpha_2$ in \cite[(4.23)]{tawe11} whose action on the
$T_{\vec{a}}, \vec{a} \in \F_2^4$, is\footnote{Here and in the following we only
list the action on those $T_{\vec a}$ which are not invariant under the respective symmetry.}
\begin{eqnarray}\label{square}
R_{G_0}(\alpha_1)\colon&&\left\{
\begin{array}{rclrclrcl}
T_{1000}&\longleftrightarrow& T_{0100},\quad&
T_{0010}&\longleftrightarrow& T_{0001},\quad&
T_{1010}&\longleftrightarrow& T_{0101},\\
T_{1001}&\longleftrightarrow& T_{0110},\quad&
T_{1110}&\longleftrightarrow& T_{1101},\quad&
T_{1011}&\longleftrightarrow& T_{0111},
\end{array}
\right.\nonumber\\[0.5em]
R_{G_0}(\alpha_2)\colon&&\left\{
\begin{array}{rclrclrcl}
T_{1000}&\longleftrightarrow& T_{0010},\quad&
T_{0100}&\longleftrightarrow& T_{0001},\quad&
T_{1100}&\longleftrightarrow& T_{0011},\\
T_{1001}&\longleftrightarrow& T_{0110},\quad&
T_{1110}&\longleftrightarrow& T_{1011},\quad&
T_{1101}&\longleftrightarrow& T_{0111}.
\end{array}
\right.
\end{eqnarray}
Similarly, 
for the non-translational symmetry group $A_4$ of the tetrahedral Kummer
surface $X_1$, we introduce three generators 
$\gamma_1,\,\gamma_2,\,\gamma_3$ in \cite[(4.23)]{tawe11}, where
\begin{eqnarray}\label{tetrahedron}
R_{G_1}(\gamma_1)\colon&&\left\{
\begin{array}{rclrclrcl}
T_{1000}&\longleftrightarrow& T_{0100},\quad&
T_{0010}&\longleftrightarrow& T_{1110},\quad&
T_{0001}&\longleftrightarrow& T_{0111},\\
T_{1001}&\longleftrightarrow& T_{0011},\quad&
T_{0101}&\longleftrightarrow& T_{1111},\quad&
T_{1011}&\longleftrightarrow& T_{1101},
\end{array}
\right.\nonumber\\[0.5em]
R_{G_1}(\gamma_2)\colon&&\left\{
\begin{array}{rclrclrcl}
T_{1000}&\longleftrightarrow& T_{0010},\quad&
T_{0100}&\longleftrightarrow& T_{1110},\quad&
T_{0001}&\longleftrightarrow& T_{1101},\\
T_{1001}&\longleftrightarrow& T_{1111},\quad&
T_{0101}&\longleftrightarrow& T_{0011},\quad&
T_{1011}&\longleftrightarrow& T_{0111},
\end{array}
\right.\\[0.5em]
R_{G_1}(\gamma_3)\colon&&\left\{
\begin{array}{rcccccccl}
T_{1000}&\longmapsto& T_{0101}&\longmapsto& T_{1101}&\longmapsto&T_{1000},\\
T_{0100}&\longmapsto& T_{0011}&\longmapsto& T_{0111}&\longmapsto&T_{0100},\\
T_{0010}&\longmapsto& T_{1001}&\longmapsto& T_{1011}&\longmapsto&T_{0010},\\
T_{0001}&\longmapsto& T_{1110}&\longmapsto& T_{1111}&\longmapsto&T_{0001},\\
T_{1100}&\longmapsto& T_{0110}&\longmapsto& T_{1010}&\longmapsto&T_{1100}.
\end{array}
\right.\nonumber
\end{eqnarray}
Finally, generators $\beta_1$ and $\beta_2$ for the triangular Kummer surface
with non-translational symmetry group $S_3$ are given in  
\cite[(1.9)]{tawe12}. Note that $\beta_2=\alpha_2$ acts as in \req{square}, 
$R_{G_2}(\beta_2)=R_{G_0}(\alpha_2)$, while
\be\label{triangle}
R_{G_2}(\beta_1)\colon\quad\left\{
\begin{array}{rcccccccl}
T_{1000}&\longmapsto& T_{0100}&\longmapsto& T_{1100}&\longmapsto&T_{1000},\\
T_{0010}&\longmapsto& T_{0011}&\longmapsto& T_{0001}&\longmapsto&T_{0010},\\
T_{1010}&\longmapsto& T_{0111}&\longmapsto& T_{1101}&\longmapsto&T_{1010},\\
T_{0110}&\longmapsto& T_{1111}&\longmapsto& T_{1001}&\longmapsto&T_{0110},\\
T_{0101}&\longmapsto& T_{1110}&\longmapsto& T_{1011}&\longmapsto&T_{0101}.
\end{array}
\right.
\ee
The permutations  $R_{G_0}(\alpha_1),\,R_{G_0}(\alpha_2),\,
R_{G_1}(\gamma_1),\,R_{G_1}(\gamma_2),\,R_{G_1}(\gamma_3),\,R_{G_2}(\beta_1)$
generate the action of $A_8\cong\GL_4(\F_2)$ on the indices $\vec a\in\F_2^4$ of the 
twisted ground states $T_{\vec a}$
in a form which is convenient for us \cite{tawe12}, but of course they
do not furnish a minimal set of generators.\\

At this point, it is important to keep in mind that we view any of our Kummer surfaces
as coming equipped with a preferred choice of generators for the lattice $\Lambda$ which defines
the underlying torus. The indexing of the twisted ground states by $\vec a\in\F_2^4$ 
is directly correlated to this choice.
As we explain in \cite{tawe11,tawe12}, this  induces 
a choice of common marking for all our Kummer surfaces, that is, an isometry between the lattice
$H_\ast(X,\Z)$ of integral homology on K3 with a standard lattice of signature $(4,20)$. 
This marking ensures that we can view the location of the Kummer lattice\footnote{For a Kummer
surface $X$ obtained by blowing up the $16$ singularities of $T/\Z_2$ for some complex
torus $T$, the \textsc{Kummer lattice} $\Pi$ is the smallest primitive sublattice of $H_\ast(X,\Z)$
which contains the classes of the sixteen rational curves obtained from the blow-up.}
$\Pi$ in $H_\ast(X,\Z)$ as fixed among all Kummer surfaces. 

Another lattice that plays an important role in our 
{previous} work is the Niemeier lattice $N$ with root lattice $A_1^{24}$.
We denote by $f_n,\, n\in\III$, with $\III=\{1,\,\ldots,\,24\}$ a choice of $24$ pairwise perpendicular
roots in $N$. The Mathieu group $M_{24}$ acts faithfully by lattice automorphisms
on the Niemeier lattice $N$, 
permuting the roots $\{f_1,\,\ldots,\,f_{24}\}$. The Niemeier lattice\footnote{For a lattice $\Gamma$,
by $\Gamma(n)$, $n\in\Z$, we denote the  $\Z$-module $\Gamma$ with quadratic form 
rescaled by the factor $n$.}
 $N(-1)$ contains a lattice
$\wt\Pi(-1)$ which is isometric to the Kummer lattice \cite[Prop.~2.3.3]{tawe11}.
In \cite[Prop.~2.3.4]{tawe11} we construct an explicit isometry $\iota\colon \Pi\longrightarrow\wt\Pi(-1)$, where
we define the lattice $\wt\Pi(-1)$ by choosing a reference octad $\OOO_9:=\{3,\,5,\,6,\,9,\,15,\,19,\,23,\,24\}$
from the Golay code,
\be\label{pidef}
\wt\Pi(-1):=\left\{ v\in N\mid \langle v,f_n\rangle = 0\;\; \fa n\in\OOO_9\right\}.
\ee
The lattice $H_\ast(X,\Z)$ is central to our description of symmetries, since  
every symmetry of a K3 surface $X$ induces a lattice automorphism on $H_\ast(X,\Z)$
which by the Torelli theorem for K3 surfaces determines the symmetry uniquely. The discussion
of symmetries hence reduces to a discussion of lattice automorphisms. In particular,
for every symmetry group $G$ of a K3 surface $X$, the action on the orthogonal complement $L_G$
of the invariant lattice, $L_G:=(H_\ast(X,\Z)^G)^\perp\cap H_\ast(X,\Z)$, uniquely determines
the action of $G$. If $X$ is a Kummer surface with induced dual K\"ahler class, then $G$ also induces a lattice automorphism
on the Kummer lattice $\Pi$, which in turn uniquely determines the action of $G$ on $X$. 
Since $L_G$ can never contain $\Pi$, and $\Pi$ in general does not contain $L_G$,
in \cite{tawe11} we propose to consider the lattice $M_G$ which is generated by $L_G$, $\Pi$,
and the vector $\upsilon_0-\upsilon$, where $\upsilon_0\in H_0(X,\Z)$ and $\upsilon\in H_4(X,\Z)$
with\footnote{On $H_\ast(X,\Z)$, we use the standard quadratic form induced by the intersection form.} 
$\langle\upsilon_0,\upsilon\rangle=1$. 

A key result is \cite[Theorem 3.3.7]{tawe11}, which in the case of Kummer surfaces with induced dual K\"ahler 
class generalizes and improves techniques introduced by Kondo \cite{ko98}. It
states that for each such Kummer surface with symmetry group $G$, the lattice $M_G$ mentioned above
can be  primitively embedded in  $N(-1)$ in such a way  that on the Kummer lattice $\Pi$, the embedding
induces the isometry $\iota\colon \Pi\longrightarrow\wt\Pi(-1)$ of \cite[Prop.~2.3.4]{tawe11}. We call
such an embedding $\iota_G \colon M_G\hookrightarrow N(-1)$ 
a \textsc{Niemeier marking}. The Niemeier marking allows us
to represent the group $G$ as a group of lattice automorphisms on the Niemeier lattice $N(-1)$, where
the action on the image of $M_G$ is prescribed by enforcing the Niemeier marking to be
$G$-equivariant, while on the orthogonal complement of $\iota_G(M_G)$, the group $G$ acts trivially. 
This allows us to elegantly realize $G$ as a subgroup of the Mathieu group $M_{24}$.
Its action on the Niemeier lattice $N(-1)$ is uniquely determined by its action on
the sublattice $\wt\Pi(-1)$.
Moreover, this construction allows us to combine symmetry groups of distinct Kummer surfaces
by means of their action on $N(-1)$. \\

As  was mentioned above, in \cite{tawe12} we show that
the combined action of all symmetry groups of Kummer K3s 
on $N(-1)$ -- and by the above, equivalently, on $\F_2^4$ -- yields the group
$(\Z_2)^4\rtimes A_8$. Here, the normal subgroup $(\Z_2)^4$ is the common translational
subgroup of all symmetry groups of Kummer surfaces, which on the labels $\vec a\in\F_2^4$ acts
by translation as in \req{translation}. This naturally fixes the action on the sublattice $\wt\Pi(-1)$
of the Niemeier lattice $N(-1)$. The translational group $(\Z_2)^4$ acts trivially on the orthogonal
complement of $\wt\Pi(-1)$ in $N(-1)$. 

The non-translational group $A_8\cong\GL_4(\F_2)$ acts on the labels 
$\vec a\in\F_2^4$ as the linear group $\GL_4(\F_2)$. In
terms of our favourite generators, this is encoded in \req{square}, \req{tetrahedron} and \req{triangle}, and
this naturally determines the action on the sublattice $\wt\Pi(-1)$ of the Niemeier lattice $N(-1)$. 
On the orthogonal complement of $\wt\Pi(-1)$ in $N(-1)$, the action is obtained from this by means
of the isomorphism $A_8\cong\GL_4(\F_2)$. The result is most conveniently described in terms of the
induced permutation of the roots in this lattice, which by \req{pidef} are labelled by our reference octad
$\OOO_9$. The permutations of the eight points of this octad that are induced by our symmetries
$\gamma_1,\,\gamma_2,\,\gamma_3,\,\alpha_1,\,\alpha_2,\,\beta_1$ are
\cite[(3.1),(3.2),(3.9)]{tawe12}
\be\label{octadaction}
\begin{array}{rclrclrcl}
\check\gamma_1 &=& (9,24)(15,19),\quad
\check\gamma_2 &=& (9,19)(15,24),\quad
\check\gamma_3 &=& (9,19,24),\\[5pt]
\check\alpha_1 &=& (6,19)(23,24),\quad 
\check\alpha_2 &=& (3,9)(23,24),\quad 
\check\beta_1 &=& (5,24,23).
\end{array}
\ee
%
%%%%%%%%%%%%%%%%%%%%%%%%%%%%%%%%%%%%%%%%%
\subsection{The proof of Proposition \ref{repona}}\label{sec:proofofprop}
%%%%%%%%%%%%%%%%%%%%%%%%%%%%%%%%%%%%%%%%%
To prove Prop.~\ref{repona}, let us first
consider the translational subgroup $(\Z_2)^4\subset\Aff(\F_2^4)$,
which acts by \req{translation} on the twisted ground states.
In Margolin's representation $M\colon\Aff(\F_2^4)\longrightarrow\End_\C(V_{45})$, 
this group is simultaneously diagonalised by the basis 
$\left\{P_X \mid X\in\{A,\,B,\,\ldots,\,N,\,O\}\right\}$ which yields the decomposition \req{45d}. 
Hence we need to use the common eigenbasis of the translational group $(\Z_2)^4$ 
on $\AAA$, which is given by
\be\label{eigenbasis}
\fa\vec a\in\F_2^4\colon\;
N_{\vec a}^{CFT} := {\textstyle{1\over4}} \sum_{\vec b\in\F_2^4} (-1)^{\langle\vec a,\vec b\rangle} T_{\vec b},
\;
\mbox{ such that }\; 
\iota_{\vec c}(N_{\vec a}^{CFT}) = (-1)^{\langle\vec c,\vec a\rangle}N_{\vec a}^{CFT}\;\;\fa \vec c\in\F_2^4,
\ee
where $\langle\cdot,\cdot\rangle$ denotes the standard scalar product on $\F_2^4$.
Hence an isomorphism of representations of $(\Z_2)^4$ between
$\AAA$ and the base $\BBB$ of $V_{45}$ is induced by identifying the translations $\iota_1,\,\ldots,\,\iota_4$
by the four standard basis vectors of $\F_2^4$ with any set of four generators of 
$(\Z_2)^4=\left\{\id,\,A^\prime,\,B^\prime,\ldots,\,N^\prime,O^\prime\right\}$ according to
Table \ref{tab:char} in the Appendix. We choose 
$$
\iota_1=A^\prime,\qquad\iota_2=B^\prime, \qquad\iota_3=D^\prime, \qquad\iota_4=F^\prime.
$$
Then from Table \ref{tab:char}  we read, for example, $\iota_1(N_A)=N_A,\,
\iota_2(N_A)=N_A,\, \iota_3(N_A)=-N_A,\,\iota_4(N_A)=-N_A$ and hence \req{eigenbasis}
implies $N_A=N_{0011}^{CFT}$. Altogether we have
\be\label{baseidentification}
\begin{array}{rclrclrclrclrcl}
N_A&=& N_{0011}^{CFT},&
N_B&=& N_{0001}^{CFT},&
N_C&=& N_{0010}^{CFT},&
N_D&=& N_{1001}^{CFT},&
N_E&=& N_{1010}^{CFT},\\[5pt]
N_F&=& N_{1110}^{CFT},&
N_G&=& N_{1101}^{CFT},&
N_H&=& N_{1000}^{CFT},&
N_I&=& N_{1011}^{CFT},&
N_J&=& N_{0111}^{CFT},\\[5pt]
N_K&=& N_{0100}^{CFT},&
N_L&=& N_{0110}^{CFT},&
N_M&=& N_{0101}^{CFT},&
N_N&=& N_{1111}^{CFT},&
N_O&=& N_{1100}^{CFT}.
\end{array}
\ee
To complete the proof of Prop.~\ref{repona}, it remains to
check that the induced action of $A_8\cong\GL_4(\F_2)$ on the orthonormal basis
$\left\{N_X\mid X\in\{A,\,B,\,\ldots,\,N,\,O\}\right\}$ of $\AAA$ {indeed yields the representation on $\AAA$ 
equivalent to the one on the
base $\BBB$ of  $V_{45}$} described in Appendix\,\ref{sec:margolin}, by 
means of the isomorphism induced by $N_X\mapsto P_X$ for all $X\in\{A,\,B,\,\ldots,\,N,\,O\}$. 

To do so,
one first calculates the permutation of $\{A,\,B,\,\ldots,\,N,\,O\}$ induced by
\req{square}, \req{tetrahedron}, \req{triangle}. For example, $\gamma_1$ 
interchanges $N_A$ and $N_C$, $N_D$ and $N_J$,\ldots, $N_I$ and $N_K$, and we write
$$
M(\gamma_1) \colon\quad (A,C)(D,J)(E,M)(G,O)(H,L)(I,K).
$$
Next, using the  array $\AAA_{even}$ of Table \ref{tab:Aeven} in the Appendix one checks that 
$s=(0,2)(1,5)$ is the unique even permutation  of the eight points $\{\infty,\,0,\,\ldots,\,6\}$,
such that conjugation by $s$ induces the permutation $(A,C)(D,J)(E,M)(G,O)(H,L)(I,K)$
of the rows of $\AAA_{even}$. We denote this by
$$
(A,C)(D,J)(E,M)(G,O)(H,L)(I,K)=\rho_{(0,2)(1,5)},
$$ 
and we proceed analogously for the other generators listed in
\req{square}, \req{tetrahedron}, \req{triangle}.
Altogether we obtain
\be\label{conjugateaction}
\begin{array}{lrcl}
M(\alpha_1) \colon&  (B,C)(D,L)(E,M)(F,G)(H,K)(I,J)&=& \rho_{(\infty,2)(1,6)},\\[5pt]
M(\alpha_2) \colon& (A,O)(B,K)(C,H)(D,L)(F,I)(G,J) &=& \rho_{(\infty,2)(0,4)},\\[5pt]
M(\gamma_1) \colon& (A,C)(D,J)(E,M)(G,O)(H,L)(I,K) &=& \rho_{(0,2)(1,5)},\\[5pt]
M(\gamma_2) \colon& (A,G)(C,O)(D,L)(E,I)(H,J)(K,M) &=& \rho_{(0,1)(2,5)},\\[5pt]
M(\gamma_3) \colon&(A,H,I)(B,N,F)(C,J,M)(D,K,G)(E,O,L)  &=& \rho_{(0,1,2)},\\[5pt]
M(\beta_1) \colon& (A,C,B)(D,N,L)(E,G,J)(F,M,I)(H,O,K) &=& \rho_{(\infty,3,2)}.
\end{array}
\ee
We can now confirm that \req{baseidentification} 
under $N_X\mapsto P_X$ for all $X\in\{A,\,\ldots,O\}$
furnishes an isomorphism between
the space $\AAA$ of twisted ground states, on the one hand, and the base $\BBB$ of $V_{45}$, on the other hand, 
as representations of 
$(\Z_2)^4\rtimes A_8$. Indeed,  by construction, this correctly identifies the
action of the translational subgroup $(\Z_2)^4$. As explained
in Section\,\ref{sec:recap},
the action of
the non-translational subgroup $A_8$ as permutation group  
is most efficiently determined by the action \req{octadaction}
of $A_8$
on the reference octad $\OOO_9$.  Finally,
the following bijection $\OOO_9\longrightarrow\{\infty,\,0,\,\ldots,\,6\}$ 
induces $\check\gamma_1\mapsto \rho_{(0,2)(1,5)},\ldots,
\check\beta_1\mapsto \rho_{(\infty,3,2)}$ and thus proves that $N_X\mapsto P_X$, $X\in\{A,\,\ldots,O\}$,
gives an isomorphism between representations of $(\Z_2)^4\rtimes A_8$:
$$
3\mapsto 4,\quad
5\mapsto 3,\quad
6\mapsto 6,\quad
9\mapsto 0,\quad
15\mapsto 5,\quad
19\mapsto 1,\quad
23\mapsto \infty,\quad
24\mapsto 2.
$$
\hspace*{\fill}$\blacksquare$
%%%%%%%%%%%%%%%%%%%%%%%%%%%%%%%%%%%%%%%%%%%%%%
\section{A twist in the $(\Z_2)^4\rtimes A_8$ action}
\label{sec:twist}
%%%%%%%%%%%%%%%%%%%%%%%%%%%%%%%%%%%%%%%%%%%%%
In the previous sections, we have constructed a $(45+\qu{45})$-dimensional space of states 
$V_{45}^{CFT}\oplus\qu V_{45}^{CFT}$, which is generic to all $\Z_2$-orbifold conformal field theories
on K3, and which accounts for the leading order term $90q$ in the function $e(\tau)$ of \req{e} that
governs the massive contributions to the elliptic genus \req{n4deco}. 
According to Prop.~\ref{our45}, this space decomposes as
$V_{45}^{CFT}=\WWW\otimes\AAA$, where $\WWW$ is a complex
$3$-dimensional vector space, while $\AAA$ is a $15$-dimensional space of 
twisted ground states which carries a faithful action of $(\Z_2)^4\rtimes A_8$. 
According to Prop.~\ref{repona}, this representation is equivalent to the representation $M$ of $(\Z_2)^4\rtimes A_8$ 
on the base $\BBB$ of the space $V_{45}=\VVV\otimes\BBB$ 
of \req{45d} which was constructed by Margolin \cite{ma93}.\\

We are now ready to explain how the properties of $V_{45}^{CFT}$  give evidence in favour of our surfing ideas, whose ultimate goal is to unravel the role of the full group $M_{24}$ in the context of  
Mathieu Moonshine, and so far provide a mathematical framework for the action of the maximal subgroup $(\mathbb Z_2)^4\rtimes A_8$.
Namely, the representation of $(\Z_2)^4\rtimes A_8$ generated by the action of
geometric symmetry groups on
$V_{45}^{CFT}$ can be identified in a natural way with the $45$-dimensional irreducible 
representation $M$ of {this maximal subgroup of} the Mathieu group $M_{24}$ described by Margolin on $V_{45}$. 
Recall however from Appendix\,\ref{sec:margolin} that the representation of $(\Z_2)^4\rtimes A_8$
on $V_{45}$ does not respect the tensor product structure $V_{45}=\VVV\otimes\BBB$ in a simple way. 
More precisely, we have the orthogonal direct decomposition 
$$ 
V_{45}=\VVV_A\oplus\VVV_B\oplus\ldots\oplus \VVV_N\oplus\VVV_O
$$
according to \req{45d}, where every $M(g)$ with $g\in(\Z_2)^4\rtimes A_8$ permutes the \textsc{fibers}
$\VVV_X$ 
of $V_{45}$,  and the induced maps $\VVV_X\longrightarrow \VVV_{M(g)(X)}$
depend non-trivially on $g$ and on $X\in\{A,\,B,\,\ldots,\,N,\,O\}$. Indeed, such a ``twist'' 
{(see Def.\ \ref{twist})} is necessary, 
since
there exists no nontrivial three-dimensional representation of $(\Z_2)^4\rtimes A_8$ that $\VVV$ could carry\footnote{The minimal dimension of a nontrivial linear 
representation of the alternating group $A_8$ is seven \cite{ccnpw85}.}.

This may appear counter-intuitive at first sight, as we have not observed a twist
 in our space of states $V_{45}^{CFT}=\WWW\otimes\AAA$. Indeed, according to 
 Prop.~\ref{our45}, the three-dimensional
space $\WWW$ is identified with the representation $\bf{3}$ of $SO(3)$
under $S\colon SO(3)\longrightarrow\End_\C(V_{45}^{CFT})$: as was 
explained in Section\,\ref{sec:generic45}, every symmetry group of a $\Z_2$-orbifold CFT 
on K3 which is induced from the underlying toroidal theory by geometric symmetry groups
acts as a subgroup of $SO(3)$ on $\WWW$ by means of the representation $\bf{3}$
of $SO(3)$. 

The key to this puzzle lies in the very groups that can occur as such
symmetry groups. As we recalled in Section\,\ref{sec:recap}, the maximal
groups in our setting are the symmetry groups $G_0=(\Z_2)^4\rtimes(\Z_2\times\Z_2)$ of the square
Kummer K3, $G_1=(\Z_2)^4\rtimes A_4$ of the tetrahedral 
Kummer K3, and $G_2=(\Z_2)^4\rtimes S_3$ of the triangular Kummer K3, 
where the common translational subgroup $(\Z_2)^4$ acts trivially on $\WWW$. In other words,
only the finite subgroups $G_T=\Z_2\times\Z_2,\, A_4$ and $S_3$ of $SO(3)$ are of relevance here, all
of which have standard nontrivial $3$-dimensional representations on $\WWW$,
induced by $G_T\subset SO(3)$, $S\colon SO(3)\longrightarrow\End_\C(\WWW)$. 

In order to understand how $V_{45}^{CFT}$
can be identified with Margolin's $V_{45}$ in a natural way, we first
need to prove that these three groups act on $V_{45}$ \textsl{without a twist} (see Def.~\ref{twist}).
We may view the twist in Margolin's representation on $V_{45}$ as yet another obstruction
for any known (orbifold) CFT on K3 to enjoy a larger geometric {symmetry.}\\

Let us briefly comment on the possibility of an action of a subgroup of
$(\Z_2)^4\rtimes A_8$ on $V_{45}=\VVV\otimes\BBB$ without a twist. 
Since the translational group $(\Z_2)^4$
acts trivially on $\VVV$, we can restrict our attention to subgroups of $A_8$, in accord with
Def.~\ref{twist}.
Recall from \cite{ma93} or from Appendix\,\ref{sec:margolin} 
that the action of $\tau\in A_8$ between any two fibers $\VVV_X$ and $\VVV_{Y}$ with $Y=M(\tau)(X)$
of $V_{45}$ is given in terms of a permutation $m_\tau^{(X,Y)}$ on the seven points 
$\{0,\,\ldots,\,6\}$ of the Fano plane $\P(\F_2^3)$. Such a permutation encodes a linear map
$\VVV_X\longrightarrow\VVV_{Y}$, because a preferred set of generators of 
the vector spaces $\VVV_X$ and $\VVV_{Y}$, namely the root vectors of the lattice 
$\Lambda_3^{b7}$, is conveniently encoded in terms of lines with marked points  in $\P(\F_2^3)$.
The precise labelling by $\{0,\,\ldots,6\}$ of the $A_8$ permutations of cycle shape $2^4$ in
the rows $X$
and $Y$ of the array $\AAA_{even}$ of Table\,\ref{tab:Aeven} 
thus corresponds to a specific choice of generators for the vector spaces $\VVV_X$ and $\VVV_{Y}$.
This in particular means that a relabelling of a row $X$ simply
amounts to a change of basis in $\VVV_X$, as long as the relabelling respects the projective
linear structure of $\P(\F_2^3)$. 
It follows that for every $\tau\in A_8$, 
there exists a labelling  of the array $\AAA_{even}$ such that $\tau$ acts without a twist
according to Def.~\ref{twist}, i.e.\ such that the permutations $m_\tau^{(X, M(\tau)(X))}$ agree
for all $X\in\{A,\,B,\,\ldots,\,N,\,O\}$. That the subgroups $G_T=\Z_2\times\Z_2,\, A_4$ and $S_3$
of $A_8$ which are relevant to our construction can act without a twist is a nontrivial claim which
we need to prove:
\begin{prop}\label{nosquaretwist}
Consider the square Kummer surface $X_0$ with symmetry group 
$G_0=(\Z_2)^4\rtimes (G_T)_0$, where $(G_T)_0=\Z_2\times\Z_2$ with generators 
$\alpha_1,\,\alpha_2$ as in \mbox{\rm\cite[(4.23)]{tawe11}},
whose action  on the base $\BBB$ of $V_{45}$ is given in \mbox{\rm\req{conjugateaction}}.
Then the group $(G_T)_0$ acts without a twist on $V_{45}$.
\end{prop}
\begin{proof}
We claim that the labelling of the array $\AAA_{even}$ of  Table\,\mbox{\rm\ref{tab:Aeven}}
exhibits no element of $(G_T)_0$ with a twist.
The proof is a straightforward calculation, where we check that for
none of the generators $\alpha_1,\,\alpha_2$, there is a twist.

From \req{conjugateaction} we read that the generator $\alpha_1$ of $(G_T)_0$
acts by the conjugation $\rho_{(\infty,2)(1,6)}$
on the rows of the array $\AAA_{even}$, and it induces the permutation $\tau_1:=(B,C)(D,L)(E,M)(F,G)(H,K)(I,J)$
of the rows. One then checks for every pair $(X,\tau_1(X))$ with
$X\in\{A,\,B,\,\ldots,\,N,\,O\}$ that the induced permutation 
$m_{\tau_1}^{(X,\tau_1(X))}$ of Fano plane labels in the rows is $(0,3)(4,5)$, independently of $X$. 
In other words, $\alpha_1$ acts without a twist.
For example, $\tau_1$ maps row $A$ into itself, where conjugation by $(\infty,2)(1,6)$
interchanges the first entry $(\infty,0)(1,5)(2,3)(4,6)$, labelled $0$, with the entry 
$(\infty,3)(0,2)(1,4)(5,6)$, labelled $3$. 

Similarly, from \req{conjugateaction} we read that the generator $\alpha_2$ of $(G_T)_0$
acts by the conjugation $\rho_{(\infty,2)(0,4)}$
on the rows of the array $\AAA_{even}$, that is by $\tau_2:=(A,O)(B,K)(C,H)(D,L)(F,I)(G,J)$. 
One then checks for every pair $(X,\tau_2(X))$  with
$X\in\{A,\,B,\,\ldots,\,N,\,O\}$ that the induced permutation 
$m_{\tau_2}^{(X,\tau_2(X))}$ of Fano plane labels in the rows is $(0,3)(1,6)$, independently of $X$. 
In other words, $\alpha_2$ acts without a twist.
\end{proof}

Since the square Kummer surface has maximal symmetry, we cannot expect the labelling 
of $\AAA_{even}$ given in Table\,\ref{tab:Aeven} to exhibit an action without twist for 
the other maximal symmetry groups $G_1,\,G_2$ of Kummer surfaces as well. Nevertheless,
we have
\begin{prop}\label{notetratwist}
Consider the tetrahedral Kummer surface $X_1$ with symmetry group 
$G_1=(\Z_2)^4\rtimes (G_T)_1$, where $(G_T)_1=A_4$ with generators 
$\gamma_1,\,\gamma_2,\,\gamma_3$ as in \mbox{\rm\cite[(4.10),(4.11)]{tawe11}},
whose action  on the base $\BBB$ of $V_{45}$ is given in \mbox{\rm\req{conjugateaction}}.
Then the group $(G_T)_1$ acts without a twist on $V_{45}$.
\end{prop}
\begin{proof}

We claim that the following relabelling yields the action of $(G_T)_1$ without a twist,
where for each row of Table\,\mbox{\rm\ref{tab:Aeven}}, we list the seven labels from left to right: 
$$
\begin{array}{rlllllll}
\bl{A}:& 0& 1 &2& 3 &4& 5& 6\\
\bl{B}:& 2 &1 &3 &0 &6& 4& 5\\
\bl{C}:&3 &5 &0& 2& 1& 4& 6\\
\bl{D}:& 6& 5& 2& 0& 1& 4 &3\\
\bl{E}:& 0& 2& 3 &1& 6 &4& 5\\
\bl{F}:& 0& 1 &5 &6& 4& 2& 3\\
\bl{G}:& 0& 1& 4& 3& 5& 2& 6\\
\bl{H}:& 0 &1& 2 &6 &4& 3 &5
\end{array}
\qquad\qquad\qquad
\begin{array}{rlllllll}
&&&\\
\bl{I}:& 0 &1& 2& 5& 4& 6& 3\\
\bl{J}:& 2& 1& 6 &5 &3 &4 &0\\
\bl{K}:& 5& 3& 2& 6& 0 &4& 1\\
\bl{L}:& 6& 4& 2& 5& 1& 0& 3\\
\bl{M}:& 4& 3& 2& 6& 0 &1& 5\\
\bl{N}:& 0& 2& 6& 5& 3& 4& 1\\
\bl{O}:& 5& 6& 1& 3& 0& 4& 2
\end{array}
$$
One checks that this relabelling respects the linear structure of each 
$\P_X(\F_2^3)$ with $X\in\{A,\,B,\,\ldots,\,N,\,O\}$.
The rest of the proof is analogous to the proof of Prop.~\ref{nosquaretwist}:
From \req{conjugateaction} one reads the action of the 
generators $\gamma_1,\,\gamma_2,\,\gamma_3$ of $(G_T)_1$
on the rows of the array $\AAA_{even}$ and checks
that the induced permutation 
$m_{\gamma_1}^{(X, M(\gamma_1)(X))}$ of Fano plane labels in the rows is $(0,2)(1,4)$, independently of $X$,
while $m_{\gamma_2}^{(X, M(\gamma_2)(X))}$ yields $(0,1)(2,4)$, independently of $X$, and
$m_{\gamma_3}^{(X, M(\gamma_3)(X))}$ yields $(0,1,2)(3,5,6)$, independently of $X$.
In other words, $(G_T)_1$ acts without a twist.
\end{proof}

The final case that we need to study works analogously:
\begin{prop}\label{notriangletwist}
Consider the triangular Kummer surface $X_2$ with symmetry group 
$G_2=(\Z_2)^4\rtimes (G_T)_2$, where $(G_T)_2=S_3$ with generators 
$\beta_1,\,\beta_2=\alpha_2$ as in \mbox{\rm\cite[(1.9)]{tawe12}},
whose action  on the base $\BBB$ of $V_{45}$ is given in \mbox{\rm\req{conjugateaction}}.
Then the group $(G_T)_2$ acts without a twist  on $V_{45}$.
\end{prop}
\begin{proof}

We work analogously to the proof of Prop.~\ref{notetratwist} and
claim that the following relabelling yields the action of $(G_T)_2$ without a twist: 
$$
\begin{array}{rlllllll}
\bl{A}:& 0& 1 &2& 3 &4& 5& 6\\
\bl{B}:& 2 &6 &3 &0 &1& 5& 4\\
\bl{C}:&3 &4 &0& 2& 6& 5& 1\\
\bl{D}:& 4& 2& 3& 6& 1& 0 &5\\
\bl{E}:& 6& 1& 2 &4& 3 &5& 0\\
\bl{F}:& 2& 0 &3 &6& 1& 4& 5\\
\bl{G}:& 1& 6& 5& 3& 2& 4& 0\\
\bl{H}:& 1 &4& 2 &5 &0& 6 &3
\end{array}
\qquad\qquad\qquad
\begin{array}{rlllllll}
&&&\\
\bl{I}:& 1 &6& 5& 2& 4& 3& 0\\
\bl{J}:& 6& 1& 2 &3 &0 &5 &4\\
\bl{K}:& 4& 1& 2& 3& 6 &5& 0\\
\bl{L}:& 5& 2& 0& 3& 1& 4& 6\\
\bl{M}:& 0& 6& 2& 5& 1 &4& 3\\
\bl{N}:& 1& 2& 4& 0& 6& 3& 5\\
\bl{O}:& 4& 1& 6& 3& 0& 5& 2
\end{array}
$$
One checks that this relabelling respects the linear structure of each 
$\P_X(\F_2^3)$ with $X\in\{A,\,B,\,\ldots,\,N,\,O\}$.
From \req{conjugateaction} one reads the action of the  
generators $\beta_1,\,\beta_2=\alpha_2$ of $(G_T)_2$
on the rows of the array $\AAA_{even}$ and checks
that the induced permutation 
$m_{\beta_1}^{(X, M(\beta_1)}(X))$ of Fano plane labels in the rows is $(0,2,3)(4,5,6)$, independently of $X$,
while $m_{\beta_2}^{(X, M(\beta_2)(X))}$ yields $(0,3)(4,5)$, independently of $X$.
In other words, $(G_T)_2$ acts without a twist.
\end{proof}

In summary, for each of the three maximal symmetry groups 
$G_k=(\Z_2)^4\rtimes (G_T)_k$ of Kummer surfaces, $k\in\{0,1,2\}$,
there exists a consistent labelling of the array $\AAA_{even}$ in Table\,\ref{tab:Aeven}, such that
no twist is exhibited for the action of
 the group $(G_T)_k$ on $V_{45}$.  We furthermore find
\begin{prop}\label{equivalence}
Consider the three maximal symmetry groups 
$G_k=(\Z_2)^4\rtimes (G_T)_k$ of Kummer surfaces, $k\in\{0,1,2\}$,
whose generators are given in \mbox{\rm\cite[(4.10), (4.11), (4.23)]{tawe11}} and 
\mbox{\rm\cite[(1.9)]{tawe12}}. For each group $G_k$,
the natural action on $V_{45}^{CFT}$, which is induced by the respective symmetries
of the states of  $\Z_2$-orbifold conformal field theories
according to \mbox{\rm\req{actiondetails}},
is equivalent to the action of $G_k$ viewed as a subgroup of $(\Z_2)^4\rtimes A_8$
in Margolin's representation $M\colon G_k\longrightarrow\End_\C(V_{45})$.
\end{prop}
\begin{proof}
Prop.~\ref{repona} implies that Margolin's representation $M$ of
$G_k$ on the base $\BBB$ of
$V_{45}=\VVV\otimes\BBB$ is equivalent to the representation
$R_{G_k}$ of $G_k$
on the base $\AAA$ of $V_{45}^{CFT}=\WWW\otimes\AAA$.
Moreover,  the translational subgroup $(\Z_2)^4$  of $(\Z_2)^4\rtimes A_8$
acts trivially both on $\WWW$ and $\VVV$ in the tensor products
$V_{45}^{CFT}=\WWW\otimes\AAA$
and $V_{45}=\VVV\otimes\BBB$. Finally, by Props.~\ref{nosquaretwist}, \ref{notetratwist}
and \ref{notriangletwist} along with Prop.~\ref{our45}, $G_k$ acts without a twist
both on $V_{45}$ and on $V_{45}^{CFT}$ for each $k\in\{0,1,2\}$. 
Hence for each $G_k,\, k\in\{0,1,2\}$, it remains to be shown that the 
representation $M$ of the subgroup $(G_T)_k$ on the fibers 
of $V_{45}^{CFT}$ is equivalent to the representation $S\colon (G_T)_k\longrightarrow\End_\C(\WWW_0)$
which is induced by $(G_T)_k\subset SO(3)$ on the fibers $\WWW\cong\WWW_0$ of $V_{45}$.\\

To do so, let us fix some notations first.
Let
\be\label{wbasis}
w_1:= i\left( \chi^1\,j_+^2+\chi^2\,j^1_+ \right)N_{0000},\quad
w_2:= \chi^1\,j^1_+N_{0000},\quad
w_3:= \chi^2\,j^2_+N_{0000}
\ee
denote a basis of $\WWW_0=\spann_\C\{ N_{0000}\}\otimes\WWW$ according to 
Prop.~\ref{our45}. 
For the generators 
$\alpha_1,\,\alpha_2$ of $(G_T)_0$,
$\gamma_1,\,\gamma_2,\,\gamma_3$ of $(G_T)_1$ and $\beta_1,\,\beta_2=\alpha_2$ of $(G_T)_2$ 
we determine the induced actions on the respective fields of the $\Z_2$-orbifold CFT on the square, the tetrahedral and
the triangular Kummer K3,
according to \cite[(1.5),(1.7),(1.9)]{tawe12}.
Here, in accord with \cite{tawe12}
the complex currents $j^k_+$, $k\in\{1,\,2\}$, whose
real and imaginary parts are the four
left-handed $U(1)$-currents generating a $U(1)^4$-symmetry in $\TTT$, are
identified with the holomorphic coordinate vector fields ${\partial\over\partial z_k}$, $k\in\{1,\,2\}$.
Hence the induced
symmetries on the free fields in $\TTT$ are given by
\begin{eqnarray*}
&&\hspace*{-1em}\alpha_1\colon
\left\{\begin{array}{rcr}
\chi^1&\longmapsto& i\chi^1,\\[3pt]
\chi^2&\longmapsto& -i\chi^2,\\[3pt] 
j_+^1& \longmapsto& i j_+^1,\\[3pt]
j_+^2 &\longmapsto&-i j_+^2
\end{array}\right\},
\;
\alpha_2\colon
\left\{\begin{array}{rcr}
\chi^1&\longmapsto& \chi^2,\\[3pt]
\chi^2&\longmapsto& -\chi^1, \\[3pt]
j_+^1& \longmapsto& j_+^2,\\[3pt]
j_+^2 &\longmapsto& -j_+^1
\end{array}\right\},
\;
\beta_1\colon
\left\{\begin{array}{rcr}
\chi^1&\longmapsto& \zeta\chi^1,\\[3pt]
\chi^2&\longmapsto& \zeta^{-1}\chi^2,\\[3pt] 
j_+^1& \longmapsto& \zeta j_+^1,\\[3pt]
j_+^2 &\longmapsto& \zeta^{-1} j_+^2
\end{array}\right\},
\\[5pt]
&&\hspace*{-1em}\gamma_1\colon
\left\{\begin{array}{rcr}
\chi^1&\longmapsto& i\chi^1,\\[3pt]
\chi^2&\longmapsto& -i\chi^2,\\[3pt] 
j_+^1& \longmapsto& i j_+^1,\\[3pt]
j_+^2 &\longmapsto& -i j_+^2
\end{array}\right\},
\;
\gamma_2\colon
\left\{\begin{array}{rcr}
\chi^1&\longmapsto& \chi^2,\\[3pt]
\chi^2&\longmapsto& -\chi^1, \\[3pt]
j_+^1& \longmapsto& j_+^2,\\[3pt]
j_+^2 &\longmapsto&  -j_+^1
\end{array}\right\},
\;
\gamma_3\colon
\left\{\begin{array}{rcr}
\chi^1&\longmapsto& {i+1\over2}(i\chi^1-\chi^2),\\[3pt]
\chi^2&\longmapsto&-{i+1\over2}(i\chi^1+\chi^2),\\[3pt] 
j_+^1& \longmapsto& {i+1\over2}(i j_+^1- j_+^2),\\[3pt]
j_+^2 &\longmapsto&  -{i+1\over2}(i j_+^1+ j_+^2)
\end{array}\right\},
\end{eqnarray*}
where as before $\zeta=e^{2\pi i/3}$. Hence we have

\ba\label{cftgenerators}
&&\hspace*{-2em}S(\alpha_1)\colon\!\!\!
\left\{\!\!\begin{array}{rcr}
w_1&\longmapsto& w_1,\\
w_2 &\longmapsto& -w_2,\\
w_3& \longmapsto& -w_3
\end{array}\!\!\right\}\!\!,
\;
S(\alpha_2)\colon\!\!\!
\left\{\!\!\begin{array}{rcr}
w_1&\longmapsto& -w_1,\\
w_2 &\longmapsto& w_3,\\ 
w_3& \longmapsto& w_2
\end{array}\!\!\right\}\!\!,
\;
S(\beta_1)\colon\!\!\!
\left\{\!\!\begin{array}{rcr}
w_1&\longmapsto& w_1,\\
w_2 &\longmapsto& \zeta^{-1} w_2,\\
w_3& \longmapsto& \zeta w_3
\end{array}\!\!\right\}\!\!.\nonumber\\
\\
&&\hspace*{-2em}S(\gamma_1)\colon\!\!\!
\left\{\!\!\begin{array}{rcr}
w_1&\longmapsto& w_1,\\[3pt]
w_2 &\longmapsto& -w_2,\\[3pt]
w_3& \longmapsto& -w_3
\end{array}\!\!\right\}\!\!,
\;
S(\gamma_2)\colon\!\!\!
\left\{\!\!\begin{array}{rcr}
w_1&\longmapsto& -w_1,\\[3pt]
w_2 &\longmapsto& w_3,\\[3pt]
w_3& \longmapsto& w_2
\end{array}\!\!\right\}\!\!,
\;
S(\gamma_3)\colon\!\!\!
\left\{\!\!\begin{array}{rcl}
w_1&\longmapsto& -w_2-w_3,\\[3pt]
w_2 &\longmapsto& {i\over2}\left(-w_1-w_2+w_3\right),\\[3pt]
w_3& \longmapsto& {i\over2}\left(w_1-w_2+w_3\right)
\end{array}\!\!\right\}\!\!,\nonumber
\ea

Note that $\gamma_2=\gamma_1^2\gamma_3\gamma_1\gamma_3^{-1}$, so henceforth
we will not continue to consider the generator $\gamma_2$. We now argue that
the respective representations are equivalent for the three relevant cases.
\begin{enumerate}
\item\textbf{The square Kummer surface $\mathbf{X_0}$ with $\mathbf{(G_T)_0=\Z_2\times\Z_2}$}

For $(G_T)_0$, we have a nontrivial
action for each nontrivial element of this group both on $\WWW_0$ and on $\VVV$. Moreover,
in both spaces, every $g\in (G_T)_0$ is represented by a unitary involution with determinant $1$. 
Hence there  is a common eigenbasis of $\VVV$
for the entire group $M\left((G_T)_0\right)$, such that  for each nontrivial  $g\in (G_T)_0$, $M(g)$
has a two-fold eigenvalue $-1$ and a simple eigenvalue $1$, as in \req{cftgenerators}.
Thus there exists an orthonormal basis $\{\wt w_1,\,\wt w_2,\wt w_3\}$ of $\VVV$ such that
$$
M(\alpha_1)\colon (\wt w_1,\,\wt w_2,\wt w_3) \longmapsto (\wt w_1,\,-\wt w_2,-\wt w_3),\quad
M(\alpha_2)\colon (\wt w_1,\,\wt w_2,\wt w_3) \longmapsto (-\wt w_1,\,\wt w_3,\wt w_2).
$$
Then $w_k\longmapsto\wt w_k$ for $k\in\{1,\,2,\,3\}$ induces an equivalence of representations of $(G_T)_0$.
\item\textbf{The tetrahedral Kummer surface $\mathbf{X_1}$ with $\mathbf{(G_T)_1=A_4}$}

For $(G_T)_1=A_4$, we know from \cite{ccnpw85} that every linear $3$-dimensional representation of
the alternating group $A_4$ is either irreducible, or it decomposes into the direct sum
of three one-dimensional representations. Moreover, there is only one equivalence class of irreducible 
linear $3$-dimensional representations of $A_4$.

By \req{cftgenerators}, there exists no common eigenvector
of $S(\gamma_1)$ and $S(\gamma_3)$, hence $\WWW_0$ cannot be the sum of three
one-dimensional representations of $S((G_T)_1)$, and thus it
carries the $3$-dimensional irreducible representation of $A_4$.
It suffices to show the same for $\VVV$.

From the proof of Prop.~\ref{notetratwist} we know that $M(\gamma_1)$ acts on the
seven points $\{0,\,\ldots,\,6\}$ of the Fano plane by means of the permutation $m_{\gamma_1}:=(0,2)(1,4)$,
while $\gamma_3$ is represented by $(0,1,2)(3,5,6)$. Hence $M(\gamma_1)$ maps the lines with marked
points $(0,023),\,(0,501),\,(0,460)$ to $(2,023),\,(2,245),\,(2,612)$, implying
that the basis $\left\{(2,0,0),\,(0,2,0),\,(0,0,2)\right\}$ of $\VVV\cong\C^3$ is mapped to 
$\left\{\pm(0,b7,b7),\,\pm(-\qu{b7},1,-1),\,\pm(\qu{b7},1,-1)\right\}$.
The signs are uniquely determined by the fact that $M(\gamma_1)$ 
is represented on $\VVV$ by a linear map with determinant one, and $m_{\gamma_1}$
fixes the point labelled $3$ in the Fano plane, such that $M(\gamma_1)$ 
permutes the three pairs of root vectors which
belong to the point frame 3 in Table\,\ref{tab:type1}. From this and by similar arguments for $\gamma_3$
one obtains the following matrix representations for the generators of $(G_T)_1$ with 
respect to the standard basis of $\C^3\cong\VVV$:
$$
M(\gamma_1)=\hf \left( \begin{array}{rrr} 0&\qu{b7}&\qu{b7}\\ b7&-1&1\\ b7&1&-1 \end{array}\right), \qquad
M(\gamma_3)=\hf \left( \begin{array}{rrr} b7&1&-1\\ 0&\qu{b7}&\qu{b7}\\ b7&-1&1 \end{array}\right).
$$
Calculating the eigenvectors of $M(\gamma_1)$, one finds a one-dimensional eigenspace
with eigenvalue $+1$ and the corresponding eigenvector $\wt v_1=(\qu{b7},1,1)^T$, which as one
immediately checks is not an eigenvector of $M(\gamma_3)$. Hence $M(\gamma_1)$ and $M(\gamma_3)$  do not have
a common eigenbasis. It follows that $\VVV$ cannot be the sum of three one-dimensional
representations of $A_4$. Hence it agrees with the irreducible $3$-dimensional representation of $A_4$.
\item\textbf{The triangular Kummer surface $\mathbf{X_2}$ with $\mathbf{(G_T)_2=S_3}$}

From the proof of Prop.~\ref{notriangletwist} we know that $M(\beta_1)$ acts on the
seven points $\{0,\,\ldots,\,6\}$ of the Fano plane by means of the permutation $(0,2,3)(4,5,6)$,
while $\beta_2$ is represented by $(0,3)(4,5)$. By a calculation similar  to the one performed
for the tetrahedral Kummer surface, 
one obtains the following matrix representations for the generators of $(G_T)_2$ with 
respect to the standard basis of $\C^3\cong\VVV$:
$$
M(\beta_1)=\hf \left( \begin{array}{rrr} 0&-\qu{b7}&\qu{b7}\\ b7&-1&-1\\ b7&1&1 \end{array}\right), \qquad
M(\alpha_2)=M(\beta_2)=\hf \left( \begin{array}{rrr} 0&\qu{b7}&-\qu{b7}\\ b7&-1&-1\\ -b7&-1&-1 \end{array}\right).
$$
Calculating the respective eigenvectors one finds that with the eigenbasis
$$
\wt w_1 = \left(\begin{array}{c}1\\0\\b7\end{array}\right),\qquad
\wt w_2 = \left(\begin{array}{c}\qu{b7}\\i\sqrt{3}\\-1\end{array}\right),\qquad
\wt w_3 = \zeta^{-1} \left(\begin{array}{c}-\qu{b7}\\i\sqrt{3}\\1\end{array}\right)
$$
of $M(\beta_1)$, $\zeta=e^{2\pi i/3}$ as above,
the isomorphism induced by $w_k\longmapsto\wt w_k$
for $k\in\{1,\,2,\,3\}$ induces an equivalence of representations of $(G_T)_2$.
\end{enumerate}
\end{proof}

From Prop.~\ref{equivalence} we now infer how to identify the
space of states $V_{45}^{CFT}$ in an arbitrary $\Z_2$-orbifold CFT $\CCC=\TTT/\Z_2$
with the representation space $V_{45}$ of $M_{24}$ in a fashion which is compatible with the
relevant group $\wt G$ of symmetries of $\CCC$. The idea is similar to the surfing procedure 
described in \cite{tawe12}.
As always, we assume that all symmetries
in $\wt G$ are induced from geometric symmetries of $\TTT$ in
a fixed geometric interpretation on some torus $T=\R^4/\wt\Lambda$.  
As detailed in \cite[Section 4]{tawe12}, for at least one $k\in\{0,\,1,\,2\}$, 
we find $\wt G\subset G_k$ along with a smooth deformation of $\Lambda_k$ into 
$\wt\Lambda$, call it $\Lambda^t$ with $t\in[0,1]$ and $\Lambda^0=\Lambda_k,\,\Lambda^1=\wt\Lambda$,
such that the linear automorphism group of each $\Lambda^t$ with $t\neq0$ is 
$\wt G_T^\prime\subset SU(2)$ where $\wt G_T=\wt G_T^\prime/\Z_2$. By Prop.~\ref{equivalence},
there exists an isomorphism from $V_{45}^{CFT}$ to $V_{45}$ which induces an equivalence
of representations of $G_k$. By construction, this isomorphism yields the desired
identification of the space of states $V_{45}^{CFT}\oplus\qu V_{45}^{CFT}$ of $\CCC$. 
Note that {on Margolin's  $V_{45}=\VVV\otimes\BBB$}, the translational subgroup 
$(\Z_2)^4$  which is \textsl{common} to all geometric symmetry groups
of Kummer K3s acts faithfully on $\BBB$ and trivially on $\VVV$. Hence its fixed point set is $\{0\}$. This implies that our 
 selection of the $90$-dimensional subspace $V_{45}^{CFT}\oplus\qu V_{45}^{CFT}$
 in  $\wh V$ (see \req{genericstates}) is in fact unique. Indeed, 
 by construction, the only fixed state of the translational $(\Z_2)^4$
 in $V_{45}^{CFT}\oplus\qu V_{45}^{CFT}$
is $0$, while the group acts trivially on the six-dimensional orthogonal complement of this space
in $\wh V$. The $90$-dimensional subspace $V_{45}^{CFT}\oplus\qu V_{45}^{CFT}$ is thus
uniquely characterized by the requirement that it carries a faithful representation of the
translational subgroup $(\Z_2)^4$.
In summary, we have shown
\begin{theorem}\label{result}
Consider a $\Z_2$-orbifold CFT $\CCC$ on K3, and let $\wt G\subset\Aff(\F_2^4)=(\Z_2)^4\rtimes A_8$ 
denote the group of those symmetries of $\CCC$
which are induced from the geometric symmetries of the underlying toroidal theory in a fixed geometric
interpretation. Then the natural representation of $\wt G$ in terms of symmetries of $\CCC$ on the 
space $V_{45}^{CFT}$ of massive states    of Prop.~\mbox{\rm\ref{our45}} is equivalent to the representation
of this group on  $V_{45}$ which is obtained by restricting 
Margolin's representation $M: M_{24}\longrightarrow\End_\C(V_{45})$ to $\wt G$.
In other words, the representation of $\wt G$ on $V_{45}^{CFT}$ can be viewed as a representation which is
induced by Margolin's representation $M: M_{24}\longrightarrow\End_\C(V_{45})$.

Moreover, within the $96$-dimensional space $\wh V$ of generic states with the appropriate quantum
numbers in $\Z_2$-orbifold CFTs on K3, the subspace singled out as $V_{45}^{CFT}\oplus\qu V_{45}^{CFT}$
is uniquely determined by the  property that the action of \textsl{any} geometric symmetry group of a 
$\Z_2$-orbifold conformal field theory is equivalent to the one induced by $M$.
\end{theorem}
This theorem encompasses the main result of the present work. Indeed, we have shown that 
on a large component of the moduli space, 
there is a $45$-dimensional subspace of the space of states, which
exists generically and which accounts for the expected \textsl{net contributions} to
the elliptic genus. We have also shown
that these states are actually uniquely characterized by the action of the symmetry
groups. 

Furthermore, our surfing 
procedure predicts that the combined action of the symmetry groups at distinct 
points of the moduli space generates the action of a subgroup of $M_{24}$; 
that this should be the case is by no means clear a priori. Not only do we 
confirm this part of our prediction, but the group that we generate is a maximal 
subgroup of $M_{24}$, which is not a subgroup of $M_{23}$, and it acts in precisely 
the predicted way. This is the first piece of evidence in the literature whatsoever 
for a trace of $M_{24}$ that is intrinsic to CFTs on K3.\\

As we recall in Section\,\ref{sec:recap}, according to  our previous work \cite{tawe12} the
images of the maximal symmetry groups $G_0,\, G_1,\,G_2$ of Kummer surfaces
under the respective representations 
$R_{G_0},\,R_{G_1},\,R_{G_2}$ altogether generate the group $\Aff(\F_2^4)\cong(\Z_2)^4\rtimes A_8$. 
By construction
this implies that the combined action of these groups on $V_{45}$ yields the representation of
$\Aff(\F_2^4)$ on $V_{45}$   induced by Margolin's representation $M$. 
By Thm.~\ref{result}, in analogy to our construction of overarching symmetry groups by means
of Niemeier markings, this procedure combines symmetry groups that are obtained at
distinct points in moduli space. Due to the twisting in the representation space $V_{45}$ it is 
not clear how to interpret an induced combined action on $V_{45}^{CFT}$ geometrically. Note for example
that, according to \req{cftgenerators}, the generators $\alpha_1$ and $\gamma_1$ of this combined group 
have the same representation $S(\alpha_1)=S(\gamma_1)$ on $\WWW$. 
On the other hand, by \req{conjugateaction}, $M(\alpha_1)$ and $M(\gamma_1)$ both fix
the label $N$, so both induce a linear map on the fiber $\VVV_N$ of $V_{45}$. However, one
checks that these maps are distinct.  Indeed, $\alpha_1$ permutes the seven points in row $N$ 
of Table\,\ref{tab:Aeven} according to $(0,3)(4,5)$ with respect to the square labelling, such that
for example the first entry of that row, $\infty0.13.24.56$, is mapped to 
$\infty4.02.15.36$. On the other hand, $\gamma_1$ permutes the seven points in row $N$ according to $(0,2)(1,4)$ with
respect to the tetrahedral labelling given in Prop.~\ref{notetratwist}, i.e.\ the first entry of that row
is mapped to $\infty2.04.16.35$. This implies that there is a nontrivial twist which is induced on $\WWW$
on transition between distinct points of the moduli space of SCFTs on K3.

It may be useful to push the analogy to the Niemeier markings of \cite{tawe12} a little bit further.
Consider a SCFT $\CCC=\TTT/\Z_2$ as before, where $\TTT$ has a geometric interpretation on the
torus $T=\R^4/\Lambda$. Let $\vec\mu_1,\,\ldots,\vec\mu_4$ denote generators of the 
lattice $\Lambda^\ast$ which by means of the Euclidean scalar product we identify as a lattice
in $\R^4\cong(\R^4)^\ast$. With $\mu^l_1,\,\ldots,\,\mu^l_4$ denoting the Euclidean coordinates of 
$\vec\mu_1,\,\ldots,\vec\mu_4$, according to \cite[(A.2)]{tawe12} we consider the fields
$$
J_k(z):= \sum_{l=1}^4 \mu^l_k {j^l(z)},\quad k\in\{1,\,\ldots,\,4\},
$$
and their superpartners $\wt\Psi_k(z)$ as building blocks to construct a lattice which 
within the (chiral, chiral) algebra of $\CCC$ plays a role
analogous to that of the integral homology of K3 within the real K3 homology. For our purposes, a slightly different
lattice might be helpful. We set
$$
\kappa_{jk}:= {\textstyle\hf}\left((\wt\Psi_j)_{1\over2} (J_k)_{1\over2} + (\wt\Psi_k)_{1\over2} (J_j)_{1\over2}\right) N_{0000},\quad 
j,\,k\in\{1,\,\ldots,\,4\}\mbox{ with } j\leq k.
$$
Let $\KKK$ denote the complex $10$-dimensional vector space with basis $\{\kappa_{jk}, j,\,k\in\{1,\,\ldots,\,4\},$ $j\leq k\}$, 
and let $K\subset\KKK$ be the lattice of rank $20$ generated over $\Z$ by the $\kappa_{jk}$ and the $i\kappa_{jk}$.
By construction, $\KKK$ contains the vector space $\WWW_0$ which yields the model fiber of $V_{45}^{CFT}$.
We have a natural action of each of our maximal symmetry groups $G_0,\, G_1,\, G_2$ on the lattice $K$
which induces the action of these groups on $\WWW_0$. 
In light
of the fact that the very representations $R_{G_k}$ of our groups $G_k$, $k\in\{0,\,1,\,2\}$, on the base $\AAA$
of $V_{45}^{CFT}$ are encoded by means of their action on $\F_2^4\cong\hf\Lambda/\Lambda$,
in other words by their description in terms of the lattice $\Lambda$,
this description of the representation 
is more natural than the one obtained by the representation 
$S\colon SO(3)\longrightarrow\End_\C(\WWW_0)$ that was mentioned in Prop.~\ref{our45}.
The combined action yields an infinite group which descends to $\GL_4(\F_2)$ when we project to 
${1\over2}K/K$.

Note that compared to the standard descriptions of the moduli space of SCFTs on K3
\cite{asmo94,nawe00}, the space $\KKK$ plays the role of the real K3 homology, and the lattice $K$ is the analog
of the integral K3 homology. Then the space $\WWW_0$ plays the role of the positive definite four-plane in 
K3 homology whose relative position with respect to the integral K3 homology determines the point in moduli
space. Indeed, the relative position of the basis vectors $w_1,\,w_2,\,w_3$ of \req{wbasis} with respect to  $K$ depends on the
moduli of our CFT $\CCC$. For example, in the SCFTs associated with the square and the tetrahedral Kummer surfaces,
respectively, the basis vector $w_3$ is given by
$$
\textstyle
\kappa_{33}-\kappa_{44} + i \kappa_{34},\qquad 
\kappa_{33}+\hf\kappa_{34} + {i\over2} ( \kappa_{34}+\kappa_{44}).
$$
That these two expressions differ is the source of the twist which we observed
above when comparing the  action of $\alpha_1$ and $\gamma_1$ on the fiber $\VVV_N$ 
of $V_{45}$. The precise meaning and interpretation of this twist clearly needs further investigation.

%%%%%%%%%%%%%%%%%%%%%%%%%%%%%%%%%%%%%%%%%%%%%%%%%%%%%%%%%%%%%%%%
\section{Conclusions}
%%%%%%%%%%%%%%%%%%%%%%%%%%%%%%%%%%%%%%%%%%%%%%%%%%%%%%%%%%%%%%%
$\Z_2$-orbifold conformal field theories on K3  provide us with a concrete framework to investigate 
the nature of the CFT states counted
 by the elliptic genus of K3 surfaces. In this paper, we have focused on the massive states 
contributing to leading order. We have taken a close look at the symmetries that act on them in an effort 
to identify signatures of the $M_{24}$ Moonshine phenomenon. Our motivation has been to carry over, in a 
field theory context, the essence of what we have recently discovered by scrutinizing the geometry of 
Kummer surfaces, which form a large class of K3 surfaces \cite{tawe11,tawe12}.
In our previous work, we considered the finite symplectic automorphism groups of Kummer surfaces 
equipped with a dual K\"ahler class induced from the underlying torus. These groups are subgroups of 
three maximal symmetry groups 
$G_0=(\Z_2)^4\rtimes (\Z_2 \times \Z_2)$, 
$G_1=(\Z_2)^4\rtimes A_4$ and $G_2=(\Z_2)^4\rtimes S_3$, and we showed that they have a 
combined action on the Niemeier lattice with root lattice $A_1^{24}$. {This yields the 
action of the combined symmetry group  $(\Z_2)^4\rtimes A_8$} of all Kummer surfaces.
In \cite{tawe12} we have shown that this group is $(\Z_2)^4\rtimes A_8 \subset M_{24}$, and 
that it is the largest group one can expect to generate on the Niemeier lattice, given the restrictions imposed. 
The Niemeier lattice may be seen as a device that provides a `memory' of the action of all finite
symplectic automorphism groups of Kummer surfaces, by accommodating the action of a group 
which is maximal in $M_{24}$ but not  contained in $M_{23}$. Of course, there is a geometric 
obstruction to any string theory propagating on a Kummer surface enjoying this combined
symmetry: the Niemeier lattice and the full integral homology lattice $H_\ast(X,\Z)$ of a Kummer 
surface $X$ have same rank but different signatures.\\

In the present work, we  impose analogous restrictions on the symmetries of $\Z_2$-orbifold conformal field theories on K3. 
We assume that these theories come with a choice of generators of the $N=(4,4)$ superconformal algebra, 
which in particular fixes the $U(1)$-currents and a preferred $N=(2,2)$ subalgebra. We furthermore 
require the symmetries to fix the superconformal algebra pointwise. These restrictions ensure that every 
symmetry preserves the conformal weights and $U(1)$-charges of every field. In \cite{tawe12} we motivate
why we restrict our attention to symmetries which are compatible with taking a large volume limit.
These restrictions imply that the symmetry groups of interest to us
are those induced  geometrically in the underlying toroidal theory 
in a fixed geometric interpretation, that is, the subgroups of $G_0, G_1$ and $G_2$. 

Our first task has been to show that there are ninety massive states 
accounting for the net contribution to the leading massive order of the elliptic genus of K3, which
organise themselves into two $45$-dimensional vector spaces with tensor product 
structure $V_{45}^{CFT}=\WWW\otimes\AAA$ and  $\qu V_{45}^{CFT}=\qu\WWW\otimes\AAA$. Here, 
 $\WWW$ and $\qu\WWW$ are the $3$-dimensional representation spaces ${\bf 3}$ and ${\bf \qu 3}$ of $SO(3)$ 
which accommodate massive fermionic excitations from the twisted sector of the theory, while $\AAA$ is a 
$15$-dimensional representation space of 
$\Aff(\F_2^4){=(\Z_2)^4\rtimes A_8}$ 
accommodating twisted ground states. The next task has been to show how closely these two 
$45$-dimensional spaces are related to the complex $45(\qu{45})$-dimensional irreducible  representations of 
$M_{24}$ constructed by Margolin \cite{ma93}. Since the groups $G_0, G_1$ and $G_2$ are all subgroups of 
$(\Z_2)^4\rtimes A_8$, which is also the combined symmetry group of all Kummer surfaces, 
it is natural to study its action on the space $V_{45}^{CFT}$. We found that the representation of 
$(\Z_2)^4\rtimes A_8$ on the space $\AAA$ is equivalent to the representation of 
$(\Z_2)^4\rtimes A_8$ constructed by Margolin on the $15$-dimensional ``base'' $\BBB$ {of 
$V_{45}=\VVV\otimes \BBB$, where $V_{45}$ carries an irreducible representation of $M_{24}$}. 
In particular, symmetry-surfing the moduli space of $\Z_2$-orbifold
CFTs $\CCC=\TTT/\Z_2$ on K3 one generates the action of the  group $\Aff(\F_2^4)$ on $V_{45}^{CFT}$
from the combined actions of $G_0, G_1$ and $G_2$.
However, the action of $(\Z_2)^4\rtimes A_8$ on the  space $V_{45}$ 
does not factorize according to this tensor product structure: a twist is necessary between fibers,
as there are no nontrivial $3$-dimensional representations of this combined symmetry  
group that the fiber $\VVV$ could carry. On the other hand, such a twist is not 
apparent in the 
CFT space $V_{45}^{CFT}$ in its natural description in terms of $\Z_2$-orbifold CFTs.  
In fact, the three maximal symmetry groups $G_k, k\in \{0,1,2\}$, act without a twist on $V_{45}$, 
as we proved. Moreover, we showed 
 that their natural action  on 
$V_{45}^{CFT}$, which is induced by the respective symmetries of CFT massive states, is equivalent 
to the action of these groups viewed as subgroups of 
$(\Z_2)^4\rtimes A_8$ in Margolin's representation. 

We have been discussing generic states in $\Z_2$-orbifold CFTs $\CCC=\TTT/\Z_2$ on K3
that account for the net
contribution to the leading massive order of the elliptic genus of K3. There
is a $96$-dimensional vector space of fermionic such states, canceling a 
contribution of $6$ generic bosonic states. We have shown that this $96$-dimensional space contains 
$V_{45}^{CFT}\oplus\qu V_{45}^{CFT}$ {as the unique $90$-dimensional subspace on which
the generic geometric symmetry group $(\Z_2)^4$ of  all $\Z_2$-orbifold CFTs $\CCC=\TTT/\Z_2$ on K3
acts faithfully}.
\\

Since we restrict ourselves to maximal symmetry groups $G_k$, finding a $\Z_2$-orbifold CFT on K3 
with more geometric symmetry is obviously impossible. We 
view the twist in Margolin's representation as another manifestation of the geometric obstruction 
to accommodate larger geometric symmetry groups, reinforcing the statement just made in the previous sentence. 
Vice versa we conjecture that for any $N=(4,4)$ SCFT on K3 whose symmetry group $\wt G$ is 
a subgroup of one of the eleven subgroups of $M_{24}$, which Mukai identifies as maximal symmetry 
groups of K3 surfaces, $\wt G$ acts without a twist on Margolin's representation on $V_{45}$. Moreover, 
assume that (1) $\CCC$ is a SCFT on K3 with geometric interpretation 
on a K3 surface with symmetry group $\wt G\subset M_{24}$, (2) the B-field of $\CCC$ in this geometric interpretation
is invariant under $\wt G$, such that $\wt G$ acts as a group of geometric symmetries of $\CCC$. Then   we expect that
there is a $45$-dimensional space of massive states $\wt V_{45}^{CFT}$ of $\CCC$ with quantum numbers
$(h,Q;\qu h,\qu Q)=({5\over4},1;{1\over4},\qu Q)$, such that $\wt G$ acts on $\wt V_{45}^{CFT}$
by means of symmetries of $\CCC$. We also expect  that this representation is equivalent to the representation
of $\wt G$ on $V_{45}$ which is  induced by Margolin's representation $M: M_{24}\longrightarrow\End_\C(V_{45})$.
If true, this provides information about states in SCFTs on K3 which nobody has been able to construct so far.\\

Finally, it would be illuminating to pin down the analog of the Niemeier markings \cite{tawe11,tawe12}, which 
were designed to bring the combined group action into light. As mentioned in the introduction, we interpret  the
representation {space} $V_{45}$ as a medium which can  collect the actions of symmetry
groups from distinct points of the moduli space and combine them to representations of larger groups, 
making its role directly comparable to that of the Niemeier lattice. Moreover, we have identified a 
$10$-dimensional complex vector space $\mathcal{K}$ and  a rank $20$ lattice $K \subset \mathcal{K}$ 
that play analogous roles to the real K3 homology and  the integral 
K3 homology that were so crucial in constructing the Niemeier markings. There are however 
interesting novel features we inherit from the representation theory of $M_{24}$, which will have to await interpretation
in a way that would lift a corner of the veil surrounding Mathieu Moonshine. 
%
%%%%%%%%%%%%%%%%%%%%%%%%%%%%%%%%%%%%%%%%%%%%%%%%%%%%%%%%%%%%%%%
\section*{Acknowledgements}
K.W. thanks Paul Aspinwall and Ron Donagi for their very helpful comments.
We thank the Heilbronn Institute and the International Centre for Mathematical Sciences in Edinburgh 
as well as the (other) organisers of the Heilbronn Day and Workshop on 
`Algebraic geometry, modular forms and applications to physics', where part of this work was done. 
A.T. thanks the University of Freiburg for their hospitality, and acknowledges a Leverhulme Research Fellowship RF/2012-335.
K.W. acknowledges an ERC Starting Independent Researcher Grant StG No. 204757-TQFT.
%%%%%%%%%%%%%%%%%%%%%%%%%%%%%%%%%%%%%%%%%%%%%%%%%%%%%%%%%%%%%%%

%
%%%%%%%%%%%%%%%%%%%%%%%%%%%%%%%%%%%%%%%%%%%%%
\appendix 
\section{A 45-dimensional vector space with $M_{24}$ action}\label{sec:margolin}
%%%%%%%%%%%%%%%%%%%%%%%%%%%%%%%%%%%%%%%%%%%%%

We review Margolin's construction of a  $45$-dimensional
irreducible representation of $M_{24}$ \cite{ma93}.
 We focus on
the features of this representation that are crucial for our present work. \\

A core component in 
the construction is a $3$-dimensional vector  space $V$ over $\mathbb{F}_2$. The associated projective 
plane $\P(V)\cong \mathbb{P}(\mathbb{F}_2^3)$ is the so-called \textsc{Fano plane}: 
it contains seven points represented by
the seven non-zero vectors of  $\mathbb{F}_2^3$. There are seven lines, each comprising three 
distinct 
points whose representatives in $\F_2^3$ together with the origin form a hyperplane. Under
addition in $\F_2^3$,  the four points of every hyperplane form a Kleinian
$4$-group, that is, a group isomorphic to $(\Z_2)^2$. 
Figure\,\ref{Fano}(a) illustrates 
\begin{figure}[h]
\begin{center}
\includegraphics[width=6.5cm,keepaspectratio]{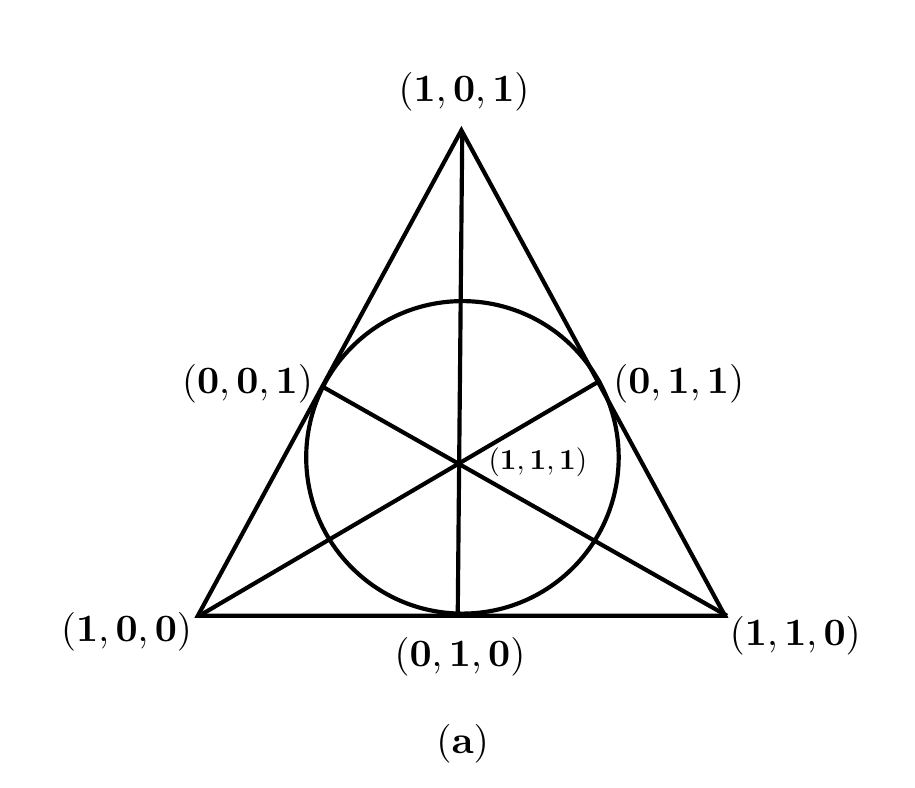}\qquad
\includegraphics[width=5.7cm,keepaspectratio]{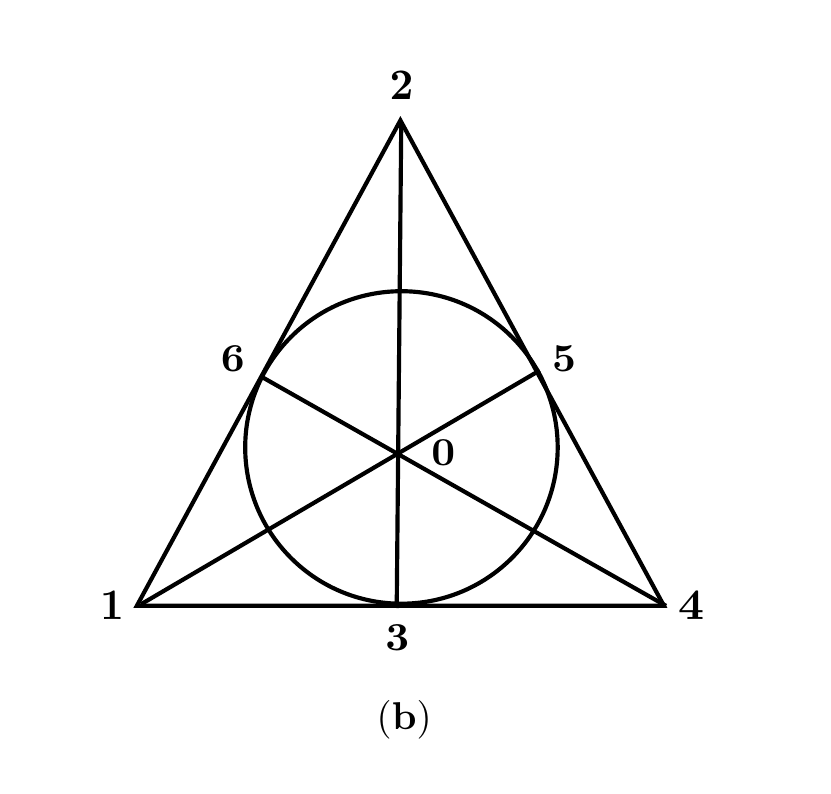}\\
% (a)\hskip 4cm (b)
\end{center}
\caption{\textsl{The Fano plane {\rm({\textsl a})} labelled by points in $\mathbb{F}_2^3$;\,\,
\mbox{\rm(}b{\rm)} labelled by points in $\mathbb{F}_7$.}}
\label{Fano}
\end{figure}
the structure of the Fano plane. 
The projective linear transformations of $\mathbb{P}(\mathbb{F}_2^3)$ act by
permuting the seven  points of  
the Fano plane, and they
form a group isomorphic to the group $L_3(2):=\GL_3(\F_2)$ of linear automorphisms  of $V$. It will be helpful for our 
purposes to label the non-zero elements of $V$ by integers modulo $7$, 
according to Figure\,\ref{Fano}(b), such that the set of points in each projective
line in $\P(V)$ has the form $\{i,\,i+2,\,i+3\},\, i\in\F_7$. In addition, the origin of $V$ is labelled $\infty$. 
Note that 
translation by a fixed vector  in  
$V$ permutes the points of $V$. In the notations of Figure\,\ref{Fano}(b), the eight resulting translations are 
\be \label{8translations}
\begin{array}{l}
{\rm translation\,\,by\,\,}(0,0,0)\,\,(\infty):\quad\quad\,\,\, \id \,\,\,\,{\rm (identity)},\\[5pt]
{\rm translation\,\,by\,\,}(1,1,1)\,\,({\rm point}\,0):\,\,(0,\infty)(1,5)(2,3)(4,6),\\[5pt]
{\rm translation\,\,by\,\,}(1,0,0)\,\,({\rm point}\,1):\,\,(1,\infty)(0,5)(2,6)(3,4),\\[5pt]
{\rm translation\,\,by\,\,}(1,0,1)\,\,({\rm point}\,2):\,\,(2,\infty)(0,3)(1,6)(4,5),\\[5pt]
{\rm translation\,\,by\,\,}(0,1,0)\,\,({\rm point}\,3):\,\,(3,\infty)(0,2)(1,4)(5,6),\\[5pt]
{\rm translation\,\,by\,\,}(1,1,0)\,\,({\rm point}\,4):\,\,(4,\infty)(0,6)(1,3)(2,5),\\[5pt]
{\rm translation\,\,by\,\,}(0,1,1)\,\,({\rm point}\,5):\,\,(5,\infty)(0,1)(2,4)(3,6),\\[5pt]
{\rm translation\,\,by\,\,}(0,0,1)\,\,({\rm point}\,6):\,\,(6,\infty)(0,4)(1,2)(3,5).
\end{array}
\ee
The Fano plane structure appears in two incarnations in
Margolin's construction, as we shall discuss now. \\

First
consider the $105$  permutations of cycle shape $2^4$
in the alternating group $A_8$, 
i.e.\ all permutations of type $(a_1,a_2)(a_3,a_4)(a_5,a_6)(a_7,a_8)$ with $a_i \in \F_7 \cup \{\infty\}$, $a_i$ all distinct.
Note that $A_8$ acts by conjugation on this set of permutations.
Denote by $\SSS$ the set of seven nontrivial translations in \eqref{8translations}, such that
$\SSS\cup\{\id\}$ is the group of translations of $V$. It follows that $\SSS\cup\{\id\}\cong(\Z_2)^3$
is the maximal normal subgroup of the group $\Aff(V) = (\Z_2)^3\rtimes L_3(2)$ of linear affine transformations of $V$.
In particular, $\SSS$ is invariant under conjugation by this group.
Since $\Aff(V)$ acts by permutations on the eight points of $V$, labelled $\infty,\,0,\,\ldots,\,6$
in \req{8translations}, we realize
this group as a subgroup of the alternating group $A_8$ of index $15$. In summary, 
the set $\SSS$ is a subset
of cardinality seven in the set of $105$  permutations in $A_8$ of cycle shape $2^4$,
which is invariant under conjugation by $\Aff(V)=(\Z_2)^3\rtimes L_3(2)\subset A_8$. It follows that 
the remaining $98$  permutations in $A_8$ of cycle shape $2^4$ decompose into 
fourteen sets of seven elements each, the images of $\SSS$ under conjugations by elements of
$A_8\setminus \Aff(V)$. 
One thus obtains an array $\AAA_{even}$ of fifteen $7$-sets such that
\begin{enumerate}
\item 
each $7$-set ${\bf X},\, X\in\{A,\,B,\,\ldots,\,N,\,O\}$, 
together with the identity permutation, forms a multiplicative abelian group of order 8;
\item 
each $7$-set ${\bf X},\, X\in\{A,\,B,\,\ldots,\,N,\,O\}$,
displays a Fano plane structure, i.e. it is possible to label the seven involutions in each set by a 
distinct element $i \in \F_7$, such that each set contains seven triplets of involutions labelled 
$\{i,i+2,i+3\}$ which correspond to the lines in the Fano plane; indeed, these triplets yield 
seven multiplicative abelian 4-groups $\{\id,i, i+2,i+3\}$. 
For instance, $\{\id,3,5,6\}$ with the involutions labelled 
$\{3,5,6\}$ in row $\bf C$ of  Table\,\ref{tab:Aeven} is the $4$-group 
$\{\id,(\infty 6)(05)(13)(24), (\infty 0)(14)(23)(56), (\infty 5)(06)(12)(34)\}$;
note that this leaves a freedom of choice for the labelling which amounts to the action of
$L_3(2)$ in each row;
\item 
the $7$-set $\SSS$ appears as row ${\bf A}$ of Table\,\ref{tab:Aeven}.
\end{enumerate}

The $45$-dimensional space affording the representation of interest 
is obtained by taking $15$ copies of a complex three-dimensional vector space to be described
below. Each copy of this vector space is labelled by one of the letters $A,\,\ldots,\,O$ corresponding
to the fifteen $7$-sets with Fano plane structure displayed in 
Table\,\ref{tab:Aeven}. This Fano plane structure on the rows
${\bf A},\,\ldots,\,{\bf O}$ will be crucial to explain the action of
$M_{24}$, or rather of its subgroup $(\Z_2)^4\rtimes A_8$, which we focus on in this appendix, 
on the representation space. Therefore we
think of the fifteen $7$-sets in Table\,\ref{tab:Aeven} as yielding the base of our representation:
formally, we introduce a  $15$-dimensional complex Euclidean vector space $\BBB$ and
choose an orthonormal basis $\{P_A,\,P_B,\,\ldots,\,P_N,\,P_O\}$ for it. It is immediate that
the vector space $\BBB$ carries a linear representation of $A_8$. Indeed,  $A_8$ acts by conjugation 
on the rows of the array $\AAA_{even}$ thus permuting the labels
$\{A,\,B,\,\ldots,\,N,\,O\}$ and thereby permuting our orthonormal basis of $\BBB$.
We refer to $\BBB$ as the \textsc{base} of Margolin's representation space.
\\

\begin{table}[htbp]
  \centering
  \tabcolsep 5.8pt
  \tiny
  \begin{tabular}{@{}clllllll@{}}
  &$c_0$&$c_1$&$c_2$&$c_3$&$c_4$&$c_5$&$c_6$\\
   \bf A&$\bf \infty 0.15.23.46^{\br 0}$&$\bf \infty 1.05.26.34^{\br 1}$&$\bf \infty 2.03.16.45^{\br 2}$&$\bf \infty 3.02.14.56^{\br 3}$&$\bf \infty 4.06.13.25^{\br 4}$&$\bf \infty 5.01.24.36^{\br 5}$&$\bf \infty 6.04.12.35^{\br 6}$\\
   &&&&&&&\\
    \bf B&$ \bf \infty 2.03.14.56^{\br 2}$&$ \infty 1.06.24.35^{\br 3}$&$\bf \infty 3.02.15.46^{\br 4}$&$\bf \infty 0.16.23.45^{\br 5}$&$ \infty 4.05.12.36^{\br 0}$&$\infty 5.04.13.26^{\br 1}$&$ \infty 6.01.25.34^{\br 6}$\\
      \bf C&$ \bf \infty 3.02.16.45^{\br 4}$&$ \infty 1.04.25.36^{\br 1}$&$\bf \infty 0.14.23.56^{\br 5}$&$\bf \infty 2.03.15.46^{\br 2}$&$ \infty 4.01.26.35^{\br 0}$&$\infty 5.06.12.34^{\br 6}$&$ \infty 6.05.13.24^{\br 3}$\\
       &&&&&&&\\
    \bf D&$  \infty 0.12.36.45^{\br 5}$&$\bf \infty 3.06.14.25^{\br 0}$&$ \infty 2.01.35.46^{\br 3}$&$\bf \infty 4.05.13.26^{\br 4}$&$ \bf\infty 1.02.34.56^{\br 6}$&$\infty 5.04.16.23^{\br 2}$&$ \infty 6.03.15.24^{\br 1}$\\
      \bf E&$  \infty 0.16.24.35^{\br 4}$&$\bf \infty 4.02.13.56^{\br 5}$&$ \infty 2.04.15.36^{\br 2}$&$\bf \infty 1.06.25.34^{\br 0}$&$ \bf\infty 3.05.14.26^{\br 1}$&$\infty 5.03.12.46^{\br 3}$&$ \infty 6.01.23.45^{\br 6}$\\
      &&&&&&&\\
\bf F&$  \infty 0.14.26.35^{\br 1}$&$\infty 1.04.23.56^{\br 0}$&$\bf  \infty 4.01.25.36^{\br 5}$&$ \infty 3.05.12.46^{\br 6}$&$ \bf \infty 5.03.16.24^{\br 2}$&$\bf \infty 2.06.13.45^{\br 4}$&$ \infty 6.02.15.34^{\br 3}$\\
\bf G&$  \infty 0.12.34.56^{\br 0}$&$\infty 1.02.35.46^{\br 1}$&$\bf  \infty 5.06.13.24^{\br 4}$&$ \infty 3.04.15.26^{\br 3}$&$ \bf \infty 2.01.36.45^{\br 5}$&$\bf \infty 4.03.16.25^{\br 2}$&$ \infty 6.05.14.23^{\br 6}$\\
&&&&&&&\\
\bf H&$  \infty 0.13.26.45^{\br 0}$&$\infty 1.03.25.46^{\br 1}$&$  \infty 2.06.15.34^{\br 2}$&$\bf \infty 5.04.12.36^{\br 6}$&$\infty 4.05.16.23^{\br 4}$&$\bf \infty 6.02.14.35^{\br 3}$&$\bf \infty 3.01.24.56^{\br 5}$\\
\bf I&$  \infty 0.16.25.34^{\br 2}$&$\infty 1.06.23.45^{\br 1}$&$  \infty 2.05.13.46^{\br 4}$&$\bf \infty 6.01.24.35^{\br 6}$&$\infty 4.03.15.26^{\br 0}$&$\bf \infty 3.04.12.56^{\br 3}$&$\bf \infty 5.02.14.36^{\br 5}$\\
&&&&&&&\\
\bf J&$\bf  \infty 4.06.12.35^{\br 6}$&$\infty 1.03.24.56^{\br 3}$&$  \infty 2.05.14.36^{\br 5}$&$\infty 3.01.26.45^{\br 1}$&$\bf \infty 6.04.15.23^{\br 0}$&$ \infty 5.02.16.34^{\br 2}$&$\bf \infty 0.13.25.46^{\br 4}$\\
\bf K&$\bf  \infty 6.04.13.25^{\br 6}$&$\infty 1.02.36.45^{\br 3}$&$  \infty 2.01.34.56^{\br 2}$&$\infty 3.05.16.24^{\br 5}$&$\bf \infty 0.12.35.46^{\br 0}$&$ \infty 5.03.14.26^{\br 1}$&$\bf \infty 4.06.15.23^{\br 4}$\\
&&&&&&&\\
\bf L&$\bf  \infty 1.05.24.36^{\br 5}$&$\bf\infty 5.01.23.46^{\br 3}$&$  \infty 2.06.14.35^{\br 0}$&$\infty 3.04.16.25^{\br 2}$&$ \infty 4.03.12.56^{\br 1}$&$ \bf \infty 0.15.26.34^{\br 6}$&$ \infty 6.02.13.45^{\br 4}$\\
\bf M&$\bf  \infty 5.01.26.34^{\br 3}$&$\bf\infty 0.15.24.36^{\br 4}$&$  \infty 2.04.13.56^{\br 2}$&$\infty 3.06.12.45^{\br 6}$&$ \infty 4.02.16.35^{\br 5}$&$ \bf \infty 1.05.23.46^{\br 1}$&$ \infty 6.03.14.25^{\br 0}$\\
&&&&&&&\\
\bf N&$  \infty 0.13.24.56^{\br 4}$&$\bf\infty 2.04.16.35^{\br 2}$&$\bf  \infty 6.05.12.34^{\br 6}$&$\infty 3.01.25.46^{\br 3}$&$ \infty 4.02.15.36^{\br 5}$&$  \infty 5.06.14.23^{ \br 0}$&$\bf \infty 1.03.26.45^{\br 1}$\\
\bf O&$  \infty 0.14.25.36^{\br 5}$&$\bf\infty 6.03.12.45^{\br 6}$&$\bf  \infty 1.04.26.35^{\br 1}$&$\infty 3.06.15.24^{\br 3}$&$ \infty 4.01.23.56^{\br 0}$&$  \infty 5.02.13.46^{\br 4}$&$\bf \infty 2.05.16.34^{\br 2}$\\
 \end{tabular}
  \caption{\textsl{Array $\AAA_{even}$ of $105$  permutations of cycle shape $2^4$ 
in $A_8$ with labelling relative to the square torus symmetry, where $a_1a_2.a_3a_4.a_5a_6.a_7a_8$
denotes the permutation $(a_1,a_2)(a_3,a_4)(a_5,a_6)(a_7,a_8).$}}
\label{tab:Aeven}
\end{table}

The Fano plane structure also appears within an irreducible $3$-dimensional representation of $L_3(2)$ constructed 
from a rank 3 sublattice $\Lambda_3^{b7}$  of the icosian Leech lattice that appears 
in Wilson's description\footnote{This
lattice is also very closely related  to the 
lattice  generated from the root graph $J_3(4)$ introduced in \cite{co76}, whose associated complex reflection group
preserves the Klein quartic 
 $F(x,y,z)=xy^3+yz^3+zy^3=0$ \cite{kl78}.} of the maximal subgroup 
$L_3(2)\rtimes\Z_2$ of the Hall-Janko group $J_2$  \cite{wi86}. 
 To understand the Fano plane structure in this context, 
 consider the following $21$ pairs of root vectors\footnote{We follow Margolin's  conventions
 and scale all root vectors to length $2$.}:
 \be\label{rootvectors}
 \pm(2,0,0)^{\sigma},\quad  \pm(0,b7,\pm b7)^{\sigma},\quad
 \pm(\pm\qu {b7},1,-1)^{\sigma}\quad \pm(\qu{b7},\pm1,\pm1)^{\sigma}
 \ee
in $\C^3$, 
 where $b7:=\hf(-1+\sqrt{-7})$ and $\qu {b7}:=\hf(-1-\sqrt{-7})$, 
 and $(a,b,c)^{\sigma}$ means that all \textsl{cyclic} permutations of $\{a,b,c\}$ should also be considered. 
 These root vectors generate the lattice $\Lambda_3^{b7}\subset\C^3$
 over $\mathbb{Z}[b7]$.  The automorphism group of this lattice is $\Z_2\times L_3(2)$, where $L_3(2)$ can be generated
 by the $21$ reflections in the root vectors \req{rootvectors}. 
 This fact can be used to define an action of $L_3(2)$ on the underlying vector space $\C^3$:
 The three root vectors $(2,0,0)^{\sigma}$ form a coordinate frame, that is an orthonormal basis of
 $\C^3$, where $(2,0,0)$ now represents the pair of root vectors $\pm(2,0,0)$, etc. The 
 remaining $18$ pairs of root vectors can be partitioned 
 into six other coordinate frames as follows:
 \begin{table}[htbp]
  \centering
 \tabcolsep 5.4pt
 \tiny
  \begin{tabular}{@{}ccccccccccccccccccccccccccc@{}}
  \multicolumn{3}{c}{\bf 0}&&\multicolumn{3}{c}{\bf 1}
  &&\multicolumn{3}{c}{$\bf 2$}&&\multicolumn{3}{c}{$\bf 3$}&&\multicolumn{3}{c}{$\bf 4$}&&
  \multicolumn{3}{c}{$\bf 5$}&&\multicolumn{3}{c}{$\bf 6$}\\
  &&&&&&&&&&&&&&&&&&&&\\
$2$& $0$ &$0$&&$-1$&$\qu{b7}$&$1$&&$-\qu{b7}$&$1$&$-1$&&$\qu{b7}$&$1$&$1$&&$b7$&$b7$&$0$&&$-b7$&$0$&$b7$&&$-1$&$-1$&$\qu{b7}$\\
$0$&$2$&$0$&&$-1$&$-\qu{b7}$&$1$&&$0$&$b7$&$b7$&&$\qu{b7}$&$-1$&$-1$&&$1$&$-1$&$\qu{b7}$&&$1$&$\qu{b7}$&$1$&&$b7$&$-b7$&$0$\\
$0$&$0$&$2$&&$b7$&$0$&$b7$&&$\qu{b7}$&$1$&$-1$&&$0$&$b7$&$-b7$&&$1$&$-1$&$-\qu{b7}$&&$-1$&$\qu{b7}$&$-1$&&$1$&$1$&$\qu{b7}$\\
 \end{tabular}
 \caption{\textsl{Point frames of the lattice $\Lambda_3^{b7}$.} }
  \label{tab:type1}
\end{table}

The labels $0,\ldots,6$ in Table\,\ref{tab:type1} indicate the structure 
of the points in a Fano plane on these seven orthonormal bases of $\C^3$: 
any pair $a,\, b$ of coordinate frames is fixed, up to signs,
by a nontrivial automorphism in $L_3(2)$, and all automorphisms fixing this pair fix a third frame
$c$. The frames $a,\, b,\, c$ then yield the points on a line in the Fano plane Figure\,\ref{Fano}(b). 
These seven frames are therefore 
called \textsc{point frames}.\\

Another partition of the pairs of root vectors into seven different coordinate frames is possible. The 
first frame consists of the first root vector pair of point frame $0$, the second root vector pair of point frame 
$2$, and the third root vector pair of point frame $3$. It is labelled $023$ and called a \textsc{line frame}. 
This way the seven line frames of Table\,\ref{tab:type2} correspond to
the lines of the Fano plane whose points are the point frames.
\begin{table}[htbp]
  \centering
 \tabcolsep 5.4pt
 \tiny
  \begin{tabular}{@{}ccccccccccccccccccccccccccc@{}}
  \multicolumn{3}{c}{\bf 023}&&\multicolumn{3}{c}{\bf 134}
  &&\multicolumn{3}{c}{$\bf 245$}&&\multicolumn{3}{c}{$\bf 356$}&&\multicolumn{3}{c}{$\bf 460$}&&
  \multicolumn{3}{c}{$\bf 501$}&&\multicolumn{3}{c}{$\bf 612$}\\
  &&&&&&&&&&&&&&&&&&&&\\
$2$& $0$ &$0$&&$-1$&$\qu{b7}$&$1$&&$-\qu{b7}$&$1$&$-1$&&$\qu{b7}$&$1$&$1$&&$b7$&$b7$&$0$&&$-b7$&$0$&$b7$&&$-1$&$-1$&$\qu{b7}$\\
$0$&$b7$&$b7$&&$\qu{b7}$&$-1$&$-1$&&$1$&$-1$&$\qu{b7}$&&$1$&$\qu{b7}$&$1$&&$b7$&$-b7$&$0$&&$0$&$2$&$0$&&$-1$&$-\qu{b7}$&$1$\\
$0$&$b7$&$-b7$&&$1$&$-1$&$-\qu{b7}$&&$-1$&$\qu{b7}$&$-1$&&$1$&$1$&$\qu{b7}$&&$0$&$0$&$2$&&$b7$&$0$&$b7$&&$\qu{b7}$&$1$&$-1$\\
 \end{tabular}
 \caption{\textsl{Line frames of the lattice $\Lambda_3^{b7}$.} }
  \label{tab:type2}
\end{table}
We can now uniquely specify every root vector in $\Lambda_3^{b7}$ 
up to sign by the line frame, that is by
the line in the Fano plane that it belongs to, along with a point on that line, that is by a point frame. 
In other words, pairs of root vectors in $\Lambda_3^{b7}$ are in $1\colon1$ correspondence with
lines in $\P(V)$ with one marked point. The 
group $L_3(2)$ of lattice automorphisms of $\Lambda_3^{b7}$  of determinant one acts faithfully on the point frames 
and on the line frames, and it  preserves the projective structure of the 
Fano plane. We obtain an induced irreducible representation of $L_3(2)$ on 
the vector space $\VVV\cong\C^3$ generated by $\Lambda_3^{b7}$ over $\C$: consider the 
orthonormal basis $(2,0,0),\,(0,2,0),\,(0,0,2)$ of $\C^3$ and specify each of these root vectors in terms of
a point and a line in the Fano plane, that is, $(0,\, 023),\;(0,501),\;(0,460)$, respectively. 
The images of the three basis vectors under  $g\in L_3(2)$
are specified, up to a sign, 
by the images of these points and lines under the permutation by which $g$
acts on the seven points of the Fano plane. The correct signs of the images follow from the
requirement that $g$ maps the pairs of root vectors in point frame $a$ to the pairs of root vectors
in point frame $g(a)$ for all $a\in\F_7$. \\

The $45$-dimensional space of the irreducible representation of $M_{24}$ that
we are interested in is obtained by assigning 
to each row of the array $\AAA_{even}$ in Table\,\ref{tab:Aeven} one of $15$ mutually orthogonal copies of 
the complex vector space $\VVV$ generated by
$\Lambda_3^{b7}$ over $\C$ \cite{ma93}. In other words, the representation space is given by 
\be \label{45d}
V_{45}:=\VVV_A\oplus\VVV_B\oplus\ldots\oplus \VVV_N\oplus\VVV_O,
\ee
where each $\VVV_X$ is a copy of $\VVV$ which carries the irreducible 
representation of $L_3(2)$ described above, 
and $X \in \{A,B,\ldots N,\,O\}$ with $A,B,\ldots,N,\,O$ labelling the rows as in Table\,\ref{tab:Aeven}.
Using the vector space $\BBB$ with orthonormal
basis $\{P_A,\,P_B,\,\ldots,\,P_N,\,P_O\}$ that was introduced above, we have
$V_{45}=\VVV\otimes\BBB$. We refer to $\BBB$ as the \textsc{base} of the {representation space} $V_{45}$, while the
$\VVV_X$ are referred to as the \textsc{fibers}.\\

It remains to identify  how the group $M_{24}$ acts on this space. 
Margolin constructs an irreducible representation $M\colon M_{24}\longrightarrow\End_\C(V_{45})$
in \cite{ma93}.
Here we  only discuss  the action of the
maximal subgroup $(\Z_2)^4\rtimes A_8$, as this is of primary relevance to our work. 
Margolin gives an explicit construction of the extra group element that generates $M_{24}$ together with the copy 
of $(\Z_2)^4\rtimes A_8$ we  describe below.\\

The group $A_8$ acts on the $45$-dimensional space $V_{45}$ as follows: 
the fifteen rows of $\AAA_{even}$ are permuted under conjugation by elements of $A_8$, so let 
$\tau \in A_8$ and $M(\tau)(X):=Y$  if $\tau{\bf X}\tau^{-1}={\bf Y}$ for rows of $\AAA_{even}$
labelled $X,Y$ ($X$ may be equal to $Y$). 
Since $\tau$ maps $4$-groups to $4$-groups, the Fano plane 
structure of the rows $\bf X$ and $\bf Y$ is preserved under $\tau$. To describe the induced action 
$\VVV_X \longrightarrow \VVV_{Y}$,
as above we use the fact that every root vector in $\Lambda_3^{b7}$ (up to a sign) is specified
by a line $p_1p_2p_3$ of the Fano plane $\mathbb{P}_X(\mathbb{F}_2^3)$ associated with row $X$  
and a point $p_k,\,k\in\{1,\,2,\,3\},$ on that line. The label of the point $p_k$ is one of  the labels attached to 
the involutions of $\AAA_{even}$ in row $X$. The permutation $\tau$ therefore induces a map 
$m_{\tau}^{(X,Y)}: \mathbb{P}_X(\mathbb{F}_2^3) \longrightarrow \mathbb{P}_Y(\mathbb{F}_2^3)$. 
The root vector image in $\VVV_{Y}$ is then reconstructed from $m_\tau^{(X,Y)}(p_k)$ and the line 
$m_\tau^{(X,Y)}(p_1)m_\tau^{(X,Y)}(p_2)m_\tau^{(X,Y)}(p_3)$.

For instance, suppose one conjugates $\AAA_{even}$ by $\tau=(\infty,0)(1,5)$. The corresponding permutation 
on the $15$ rows is given by $(B,C)(D,O)(E,N)(F,H)(G,I)(J,K)$, i.e.\ $\tau$ maps $\VVV_A$, $\VVV_L$ and 
$\VVV_M$  to themselves, $\VVV_B$ to $\VVV_C$, $\VVV_D$ to $\VVV_O$ etc. 
To determine the precise action on these spaces, one reads off the permutation on the seven points of the 
Fano plane encoded in the labelling of involutions within $\AAA_{even}$. In the case of the 
labelling displayed in Table\,\ref{tab:Aeven}, the permutation $(B,C)$ corresponds to the map 
$m_\tau^{(B,C)}: \mathbb{P}_B(\mathbb{F}_2^3)\longrightarrow \mathbb{P}_C(\mathbb{F}_2^3)$ with 
$m_\tau^{(B,C)}(0)=1,\, m_\tau^{(B,C)}(1)=0,\, m_\tau^{(B,C)}(2)=4,\,m_\tau^{(B,C)}(3)=3,\,
m_\tau^{(B,C)}(4)=2,\,m_\tau^{(B,C)}(5)=5,\, m_\tau^{(B,C)}(6)=6$. We
encode this map succinctly and mnemonically as the  ``permutation'' 
$m_\tau^{(B,C)}=(1,0)(2,4)$, which governs how  
$\VVV_B$ is mapped to  $\VVV_C$. Specifically, the root vector in $\VVV_B$ 
corresponding to point $0$ within the line $023$ of  $\mathbb{P}_B(\mathbb{F}_2^3)$ is mapped on 
$\VVV_C$ to the root vector $1$ within line $134$ of $ \mathbb{P}_C(\mathbb{F}_2^3)$. 
In other words, Margolin's representation $M\colon A_8\longrightarrow\End_\C(V_{45})$ enforces
$M( (\infty,0)(1,5) ) \left( (2,0,0)_B \right) = \pm (-1,\qu{b7},1)_C$.
Since $A$ is fixed under $\tau$, the corresponding permutation of $\mathbb{P}_A(\mathbb{F}_2^3)$ is 
$m_\tau^{(A,A)}=(2,3)(4,6)$, which induces a map from $\VVV_A$ to itself, and so on. \\

We note that the maps $m_\tau^{(X,Y)}$ and $m_\tau^{(U,V)}$ for two pairs of rows  need not  be identical. 
This prompts us to introduce the following 
\begin{definition}\label{twist} 
Let $\AAA_{even}$ be an array as in Table\,\ref{tab:Aeven}, with fixed labelling, and  
let $\tau \in A_8$ act by conjugation 
on the $15$ rows of $\AAA_{even}$. If the induced maps $m_\tau$ between Fano planes associated 
with the permuted rows are not all identical, then the array is said to exhibit
the action of $\tau$ on the $45$-dimensional space 
$V_{45}$ with a \textsc{twist}.

For a subgroup $G\subset A_8$, assume that there exists a labelling of the array 
$\AAA_{even}$ which is compatible with the linear structure of the Fano planes 
$\P_X(\F_2^4)$ for all $X\in\{A,\,B,\,\ldots,\,N,\,O\}$ and which
exhibits a twist for \textsl{no} $\tau\in G$. In other words, the action of $G$ 
factorizes according to $V_{45}=\VVV\otimes\BBB$.
Then we say that $G$ acts \textsc{without a twist}.
\end{definition}
Another partition of the $105$  permutations  of cycle shape $2^4$ in $A_8$ is useful 
for exhibiting the action of the normal subgroup $(\Z_2)^4$ in $(\Z_2)^4\rtimes A_8 \subset M_{24}$,
that is, to exhibit $M\colon (\Z_2)^4\longrightarrow\End_\C(V_{45})$. 
This action is obtained from the array $\AAA_{even}$ by conjugation with an element 
of $S_8\setminus A_8$. We choose this element to be $(0,\infty)$ and call the conjugate array $\AAA_{odd}$, 
displayed in Table\,\ref{tab:Aodd}. We label each
of the permutations in the array by a letter $A,\, B,\, \ldots,\, N,$ or $O$, according to the row in which
this permutation occurs in $\AAA_{even}$.
Note that under conjugation by $(0,\infty)$, the seven boldfaced involution 
of column $c_1$ of $\AAA_{even}$ are interchanged with the seven involutions of row $A$, leaving 
$(\infty,0)(1,5)(2,3)(4,6)$ invariant. Similarly, $(0,\infty)$ interchanges the seven boldfaced 
involutions of columns $c_i,\, i=2,\ldots,6$ with the involutions of rows $M,\, C,\, B,\, K,\, L,\, J$ respectively, leaving 
$(\infty,0)(1,5)(2,4)(3,6), (\infty,0)(1,4)(2,3)(5,6)$, $(\infty,0)(1,6)(2,3)(4,5)$,
$(\infty,0)(1,2)(3,5)(4,6)$, $(\infty,0)(1,5)(2,6)(3,4)$, and $(\infty,0)(1,3)(2,5)(4,6)$ invariant. 
In fact, each row $X^\prime$ with $X^\prime\in\{A{^\prime},\, B{^\prime},\, \ldots,\, N{^\prime},\, O{^\prime}\}$
of $\AAA_{odd}$ contains involutions from seven different rows of $\AAA_{even}$, 
and therefore specifies seven of the $15$ copies of $\VVV$. To each row $X^\prime$ of $\AAA_{odd}$,
Margolin associates an automorphism  
of $V_{45}$ (call it $X^\prime$ as 
well) that fixes  these seven copies of $\VVV$ pointwise, and
acts as $-\id$ on the other eight. In other words, $X^\prime$ acts
linearly on the base $\BBB$ of $V_{45}=\VVV\otimes\BBB$; 
it fixes the seven vectors $P_Y$ in our orthonormal
basis of $\BBB$ which have labels $Y$ occurring in the line $X^\prime$ of the dual array $\AAA_{odd}$,
while multiplying the other eight basis vectors by $-1$. $X^\prime$ acts trivially on  $\VVV$.
The resulting automorphisms $A^\prime,\, B^\prime,\, \ldots,\, 
N^\prime,\, O^\prime$
are captured by Table\,\ref{tab:char}, which can easily be identified 
as the character table of an abelian  group of order $2^4$. Since this group is normalised 
by the action of $A_8$, altogether one obtains an action of 
$(\Z_2)^4\rtimes A_8$ on the $45$-dimensional space $V_{45}$ of
\eqref{45d} and thus $M\colon (\Z_2)^4\rtimes A_8\longrightarrow\End_\C(V_{45})$. 
The representation of the full group $M_{24}$ on
$V_{45}$  is described in \cite{ma93}. \\

\begin{table}
  \centering
  \tabcolsep 5.8pt
  \tiny
 \begin{tabular}{@{}llllllll@{}}
  &$c_0'$&$c_1'$&$c_2'$&$c_3'$&$c_4'$&$c_5'$&$c_6'$\\
  $ \bf A'$&$\bf \infty 0.15.23.46^{\br A}$&$\bf \infty 5.01.26.34^{\br M}$&$\bf \infty 3.02.16.45^{\br C}$&$\bf \infty 2.03.14.56^{\br B}$&$\bf \infty 6.04.13.25^{\br K}$&$\bf \infty 1.05.24.36^{\br L}$&$\bf \infty 4.06.12.35^{\br J}$\\
   &&&&&&&\\
   $ \bf B'$&$ \bf \infty 3.02.14.56^{\br A}$&$\bf  \infty 6.01.24.35^{\br I}$&$\bf \infty 2.03.15.46^{\br C}$&$\bf \infty 0.16.23.45^{\br B}$&$ \bf \infty 5.04.12.36^{\br H}$&$\bf \infty 4.05.13.26^{\br D}$&$\bf \infty 1.06.25.34^{\br E}$\\
     $ \bf C'$&$ \bf \infty 2.03.16.45^{\br A}$&$\bf \infty 4.01.25.36^{\br F}$&$\bf \infty 0.14.23.56^{\br C}$&$\bf \infty 3.02.15.46^{\br B}$&$\bf \infty 1.04.26.35^{\br O}$&$\bf\infty 6.05.12.34^{\br N}$&$ \bf\infty 5.06.13.24^{\br G}$\\
      &&&&&&&\\
   $ \bf D'$&$  \infty 0.12.36.45^{\br D}$&$ \infty 6.03.14.25^{\br M}$&$\infty 1.02.35.46^{\br G}$&$ \infty 5.04.13.26^{\br B}$&$  \infty 2.01.34.56^{\br K}$&$ \infty 4.05.16.23^{\br H}$&$ \infty 3.06.15.24^{\br O}$\\
     $  E'$&$  \infty 0.16.24.35^{\br E}$&$ \infty 2.04.13.56^{\br M}$&$ \infty 4.02.15.36^{\br N}$&$ \infty 6.01.25.34^{\br B}$&$ \infty 5.03.14.26^{\br K}$&$\infty 3.05.12.46^{\br F}$&$ \infty 1.06.23.45^{\br I}$\\
 &&&&&&&\\
   $ \bf F'$&$  \infty 0.14.26.35^{\br F}$&$ \infty 4.01.23.56^{\br O}$&$\infty 1.04.25.36^{\br C}$&$ \infty 5.03.12.46^{\br E}$&$  \infty 3.05.16.24^{\br K}$&$ \infty 6.02.13.45^{\br L}$&$ \infty 2.06.15.34^{\br H}$\\
     $  G'$&$  \infty 0.12.34.56^{\br G}$&$ \infty 2.01.35.46^{\br D}$&$ \infty 6.05.13.24^{\br C}$&$ \infty 4.03.15.26^{\br I}$&$ \infty 1.02.36.45^{{\br 5}, K}$&$\infty 3.04.16.25^{\br L}$&$ \infty 5.06.14.23^{\br N}$\\
     &&&&&&&\\
   $ \bf H'$&$  \infty 0.13.26.45^{\br H}$&$ \infty 3.01.25.46^{\br N}$&$\infty 6.02.15.34^{\br F}$&$ \infty 4.05.12.36^{\br B}$&$  \infty 5.04.16.23^{\br D}$&$ \infty 2.06.14.35^{\br L}$&$ \infty 1.03.24.56^{\br J}$\\
     $  I'$&$  \infty 0.16.25.34^{\br I}$&$ \infty 6.01.23.45^{\br E}$&$ \infty 5.02.13.46^{\br O}$&$ \infty 1.06.24.35^{\br B}$&$ \infty 3.04.15.26^{\br G}$&$\infty 4.03.12.56^{\br L}$&$ \infty 2.05.14.36^{\br J}$\\
      &&&&&&&\\
   $ \bf J'$&$ \bf \infty 6.04.12.35^{\br A}$&$ \bf \infty 3.01.24.56^{\br H}$&$ \bf\infty 5.02.14.36^{\br I}$&$ \bf \infty 1.03.26.45^{\br N}$&$   \bf\infty 4.06.15.23^{\br K}$&$ \bf \infty 2.05.16.34^{\br O}$&$ \bf \infty 0.13.25.46^{\br J}$\\
     $  K'$&$   \bf\infty 4.06.13.25^{\br A}$&$ \bf \infty 2.01.36.45^{\br G}$&$ \bf \infty 1.02.34.56^{\br D}$&$ \bf \infty 5.03.16.24^{\br F}$&$ \bf \infty 0.12.35.46^{\br K}$&$ \bf\infty 3.05.14.26^{\br E}$&$  \bf\infty 6.04.15.23^{\br J}$\\
     &&&&&&&\\
   $ \bf L'$&$ \bf \infty 5.01.24.36^{\br A}$&$ \bf \infty 1.05.23.46^{\br M}$&$ \bf\infty 6.02.14.35^{\br H}$&$ \bf \infty 4.03.16.25^{\br G}$&$   \bf\infty 3.04.12.56^{\br I}$&$ \bf \infty 0.15.26.34^{\br L}$&$ \bf \infty 2.06.13.45^{\br F}$\\
     $  M'$&$   \bf\infty 1.05.26.34^{\br A}$&$ \bf \infty 0.15.24.36^{\br M}$&$ \bf \infty 4.02.13.56^{\br E}$&$ \bf \infty 6.03.12.45^{\br O}$&$ \bf \infty 2.04.16.35^{\br N}$&$ \bf\infty 5.01.23.46^{\br L}$&$  \bf\infty 3.06.14.25^{\br D}$\\
     &&&&&&&\\
   $ \bf N'$&$  \infty 0.13.24.56^{\br N}$&$ \infty 4.02.16.35^{\br M}$&$\infty 5.06.12.34^{\br C}$&$ \infty 1.03.25.46^{ \br{H} }$&$  \infty 2.04.15.36^{\br E}$&$ \infty 6.05.14.23^{\br G}$&$ \infty 3.01.26.45^{\br J}$\\
     $  O'$&$  \infty 0.14.25.36^{\br O}$&$ \infty 3.06.12.45^{\br M}$&$ \infty 4.01.26.35^{\br C}$&$ \infty 6.03.15.24^{\br D}$&$ \infty 1.04.23.56^{\br F}$&$\infty 2.05.13.46^{\br I}$&$ \infty 5.02.16.34^{\br J}$\\
      &&&&&&&\\
 \end{tabular}
  \caption{\textsl{Array $\AAA_{odd}$ 
  obtained from $\AAA_{even}$ through conjugation by $(0,\infty)$.} }
  \label{tab:Aodd}
  \end{table}

\begin{table}[htbp]
  \centering
  \tabcolsep 5.8pt
  \small
  \begin{tabular}{@{}c|c|ccccccccccccccc@{}}
 1 &&1&1&1&1&1&1&1&1&1&1&1&1&1&1&1\\\hline
 &&&&&&&&&&&&&&&&\\
&&$A'$&$B'$&$C'$&$D'$&$E'$&$F'$&$G'$&$H'$&$I'$&$J'$&$K'$&$L'$&$M'$&$N'$&$O'$\\\hline
&&&&&&&&&&&&&&&&\\
1&A&1&1&1&-1&-1&-1&-1&-1&-1&1&1&1&1&-1&-1\\
1&B&1&1&1&1&1&-1&-1&1&1&-1&-1&-1&-1&-1&-1\\
1&C&1&1&1&-1&-1&1&1&-1&-1&-1&-1&-1&-1&1&1\\
1&D&-1&1&-1&1&-1&-1&1&1&-1&-1&1&-1&1&-1&1\\
1&E&-1&1&-1&-1&1&1&-1&-1&1&-1&1&-1&1&1&-1\\
1&F&-1&-1&1&-1&1&1&-1&1&-1&-1&1&1&-1&-1&1\\
1&G&-1&-1&1&1&-1&-1&1&-1&1&-1&1&1&-1&1&-1\\
1&H&-1&1&-1&1&-1&1&-1&1&-1&1&-1&1&-1&1&-1\\
1&I&-1&1&-1&-1&1&-1&1&-1&1&1&-1&1&-1&-1&1\\
1&J&1&-1&-1&-1&-1&-1&-1&1&1&1&1&-1&-1&1&1\\
1&K&1&-1&-1&1&1&1&1&-1&-1&1&1&-1&-1&-1&-1\\
1&L&1&-1&-1&-1&-1&1&1&1&1&-1&-1&1&1&-1&-1\\
1&M&1&-1&-1&1&1&-1&-1&-1&-1&-1&-1&1&1&1&1\\
1&N&-1&-1&1&-1&1&-1&1&1&-1&1&-1&-1&1&1&-1\\
1&O&-1&-1&1&1&-1&1&-1&-1&1&1&-1&-1&1&-1&1\\
 \end{tabular}
  \caption{\textsl{The action of $(\Z_2)^4$ on} the 
  $45$-dimensional space $V_{45}$. }
  \label{tab:char}
\end{table}
\newpage
\newcommand{\etalchar}[1]{$^{#1}$}
\def\polhk#1{\setbox0=\hbox{#1}{\ooalign{\hidewidth
  \lower1.5ex\hbox{`}\hidewidth\crcr\unhbox0}}} \def\cprime{$^\prime$}
\providecommand{\bysame}{\leavevmode\hbox to3em{\hrulefill}\thinspace}

\end{document}